\documentclass{article} 
\usepackage[english]{babel}

\usepackage{lmodern}
\usepackage[top=2cm, bottom=1.8cm,left=2.5cm,right=2.5cm,head=2cm,headsep=0.5cm]{geometry} % needed for page border settings

\usepackage{bbm}
\usepackage{amsmath}		
\usepackage{amssymb, amsbsy,dsfont}		
\usepackage{mathrsfs}
\usepackage{latexsym}	
\usepackage{array}
\usepackage{multirow}
\usepackage{amsthm, array, braket, color}%,showkeys}
\usepackage{graphicx}
\usepackage{hyperref}
\usepackage{verbatim}
\usepackage{authblk}
\usepackage{tikz}
\usetikzlibrary{matrix}
\usepackage{csquotes}

\usepackage[toc,page]{appendix} 
\usepackage{enumitem}

\renewcommand{\epsilon}{\varepsilon}

\newcommand{\E}{\mathbb{E}}

\newcommand{\R}{\mathbb{R}}

\newcommand{\C}{\mathbb{C}}
\newcommand{\N}{\mathbb{N}}
\renewcommand{\P}{\mathbb{P}}

\newcommand{\Un}{\mathds{1}}

\newcommand{\norm}[1]{{\left\|{#1}\right\|}}
\newcommand{\ent}[1]{{\left[{#1}\right]}}
\newcommand{\abs}[1]{{\left|{#1}\right|}}
\newcommand{\scal}[1]{{\left\langle{#1}\right\rangle}}
\newcommand{\plus}{{\pmb{+}}}
\newcommand{\minus}{{\pmb{-}}}

\newcommand{\dd}{\mathcal{D}}

\newcommand{\pp}{\mathcal{P}}
\newcommand{\hh}{\mathcal{H}}

\newcommand{\qq}{\mathcal{Q}}
\newcommand{\call}{\mathcal{L}}
\newcommand{\bb}{\mathcal{B}}

\newcommand{\yy}{\mathcal{Y}}

\newcommand{\eps}{\epsilon}

\newcommand{\gs}{\mathfrak{S}}

\newcommand{\tr}[1]{\text{Tr}\left( {#1} \right) }
\newcommand{\tra}{\text{Tr} }

\DeclareMathOperator{\cov}{Cov}
\DeclareMathOperator{\ran}{ran}

\newtheorem{theo}{Theorem}
\newtheorem{prop}[theo]{Proposition}
\newtheorem{defi}[theo]{Definition}
\newtheorem{lem}[theo]{Lemma}

\theoremstyle{definition}
\newtheorem{rem}{Remark}

\newtheorem{assumption}{Assumption}

\begin{document}
\title{Energy exchange and entropy for quasi-free fermionic semigroups}
\author{Simon Andr\'eys}
\maketitle

\begin{abstract}
We consider a model of quantum dynamical semigroup on a finite dimensional fermionic space, obtained as the continuous-time limit of a repeated interactions model between a system and several thermal baths, with a dynamic driven by quadratic Hamiltonians. We assume that there is a globally conserved observable which can be expressed as a sum of energies on the system and on each baths, and we study the energy fluxes between the baths and the system. First, we consider only the mean energy fluxes, and prove that every thermal machines on quasi-free fermions in trivial, in the sense that it is not possible to extract energy from the coolest bath, even when we dispose of several other baths at different temperatures. Then, we consider an unraveling of the semigroup as a random process, and we study the large deviations of the energy fluxes, following \cite{JPW14}. We reduce the computation of the cumulant generating functional to the resolution of a Riccati equation (which is formally similar to the study of large deviations in classical networks of harmonic oscillators \cite{JPS16}). We apply it to the numerical computation of the rate function for energy exchanges in a fermionic chain, and show that larger fluctuations are observed on a longer chain. 
\end{abstract}

\tableofcontents

\section{Introduction}

At the crossroad of quantum dynamics and thermodynamics, the theory of energy exchanges between quantum systems has attracted a lot of interest. Several articles focused on the fluctuations of energy and entropy fluxes, in equilibrium and non-equilibrium steady states, either in the two-time measurement protocol \cite{Kurchan00},\cite{EHM_nonequilibrium_2009},\cite{BJPPP15},\cite{BPP18},\cite{BPR19} or for continuous-time measurements \cite{DerezinskiRoeckMaes08_fluctuations},\cite{JPW14}. We may also cite the lecture notes \cite{JOPP_entropic_2012}, which contains a detailed treatment of Fermionic systems, and the article \cite{JPPP_energy_2015}, which does not consider fluctuations but energy conservations in the two-time measurement protocol. The definition of the fluctuations in classical systems is well-established and linked with large deviations (note the article \cite{BJP_energy_2017} with a comparison with experimental data), while there are several quantum analogues for the fluctuations of currents, as noted in \cite{DerezinskiRoeckMaes08_fluctuations}; the abstract way is to study a function $e(\alpha)$ linked with some Renyi entropy, as in \cite{JLP_fluctuations_2013}. The most mundane way (and the one of this article) is to study the large deviations of a random variable obtained by measuring some energy observables on the systems. Besides the study of fluctuations, we may also note some works on the Landauer principle (\cite{Landauer_reebWolf},\cite{JP_landauer_2014},\cite{BFJP16} \cite{HJPR18}) or in link with resource theory (\cite{fully_quantum_second_law_CSHO}, \cite{thermal_operations_Mazurek18}).

A convenient approach to this subject is the one of Markovian dynamics: we consider a small system $S$ in contact with an exterior system $B$, and assume that the effect of the exterior systems on $S$ at a time $t$ does not depends on the past interactions between them.
Then, the density matrix on the system $S$ evolves according to a linear equation; the dynamic is described by a so-called quantum Markov semigroup, whose generator is a linear super-operator $\call$ called Lindbladian \cite{Lindblad1976} or GKSL operator.
This type of evolution are used as effective models for open systems, notably in quantum optics \cite{optics_gardiner_zoller_qt_noise}, and they where considered since the beginnings of quantum mechanics (see Landau \cite{landau27} or its English translation in \cite{landau_collected}), though their general theory really started in the '70 (see \cite{briefHistory} for a history).
A problem is to rigorously derive a Markovian evolution from a standard Hamiltonian evolution on complete system $SB$; it may be obtained as some limit for a convenient scaling, which is sometimes called a stochastic limit (\cite{accardi_stochastic_limit}, \cite{accardiLu_first_40_years_GKSL}) and encompasses the low density limit and most importantly the weak coupling limit (\cite{Davies74} \cite{Derezinski2007}).
The fluctuation  of currents in models obtained from the weak coupling limit has notably been studied in \cite{DerezinskiRoeckMaes08_fluctuations}, \cite{RoeckMaes06_fluctuations}.

 In this article, we use another derivation of Markovian dynamics, the continuous-time limit of repeated interactions (\cite{AttalPautrat}, \cite{AttalJoye}, \cite{repeated_interactions_BAM14}). This allows to make a direct relation between discrete-time dynamics and the quantum Markov semigroup. In the repeated interaction framework, the exterior system is divided as the sum of identical subsystems which interact one after the other with the system $S$ during a time $\tau$. As $\tau \rightarrow 0$ and under suitable normalization of the Hamiltonian, the obtained dynamics converges to a continuous-time Markovian semigroup $(\Lambda^t)_{t\in [0, +\infty)}$ . Importantly, it is possible to measure some observable before and after each measurements, making the evolution a random process. The limit evolution is then a stochastic process $(\tilde{\rho_t})_{t\in [0, +\infty)}$ which is linked with the semigroup $(\Lambda^t)_{t\in [0, +\infty)}$ through the relation
 \[
 \E(\tilde{\rho_t})=\Lambda^t(\rho_0)~.
 \] 
 and is called an unraveling of the semigroup(\cite{srinivas_davies_81}\cite{Belavkin_07_eventum}\cite{Carmichael_lecture93}\cite{Breuer2002}). The convergence in distribution of the process has notably been studied by Pellegrini \cite{PellegriniDiffusive08}\cite{Pellegrini_Nechita09}\cite{pellegrini_jumps_10}.
 
Under some assumptions of detailed balance on the Lindbladian and its unraveling, it is possible to interpret the measured observables as energy exchanged between the systems and some subsystems $B_1, \cdots, B_n$ of $B$, and to study their large deviations. This is the subject of the article \cite{JPW14}, which is the main inspiration of this article. In this context, the Lindbladian of the system is expressed as 
\[
\call=i[H_S, \bullet]+\sum_{i=1}^n \call_i
\]
where the $\call_i$ represents the effect of the $i$-th bath, and there is some pseudo-potential energy $K_S$ which commutes with $H_S$ and such that the $\call_i$'s satisfy a detailed balance condition with respect to the Gibbs state 
\[
\sigma_{\beta_i}=\frac{e^{-\beta_i K_S}}{\tr{e^{-\beta_i K_S}}}~.
\]
It is then possible to define the mean energy flux entering the $i$-th bath as 
\[
J_i=-\tr{\call_i(K_S)\rho_\infty}
\]
and some random processes $(N^i_t)_{t\in [0, +\infty)}$ where $N_t^i$ represents the energy increase measured in the $i$-th bath, with the relation
\[
\lim_{t\rightarrow \infty} \frac{1}{t} \E\left(N_t^i\right)=J_i~.
\]

The contributions of the present work are first, to explore the link between the  continuous-time and the discrete-time approaches to the fluctuations of energy fluxes through the limit of repeated interaction; then, to study both the mean energy fluxes $J_i$ and the large deviations of the $N_t^i$ in the particular case of quasi-free fermionic systems.

A free fermionic dynamics describes non-interacting fermions which may jump between the system and the baths, of which quasi-free fermionic semigroups is a generalization. Quasi-free fermionic dynamics on bosonic and fermionic spaces has long been studied (see for example \cite{AlickiLendi} or \cite{fagnolaRebolledo2002}). These models have the advantage to be explicitly solvable in many cases, and still show non-trivial behavior which make them good toy models, for example to test quantum functional inequalities as in \cite{TPK14}. Moreover, they are formally similar to classical networks of harmonic oscillators driven by Langevin noise \cite{MNV03_heat_network}\cite{EckmannZabey04}\cite{JPS16}. Some recent works of Prosen \cite{Prosen2008} \cite{Prosen2010} introduced methods to study the convergence as $t\rightarrow +\infty$ to a unique stationary state, and the author of the present article established a necessary criterion for the convergence and uniqueness \cite{AndreysFermions1}. The repeated interaction model on fermionic spaces was introduced by Platini and Karevski \cite{Platini2008} \cite{KarevskiPlatini2009} in the case of the XY model.
\vspace{0.5cm}

In the second section of this article, we concentrate on the study of the mean energy fluxes $J_i$. We describe a general framework of repeated interaction models with a globally conserved quantity $K$ (which is for example used in \cite{HJPR18}), and show (Proposition \ref{prop:thermal=db}) how the conservation of $K$ corresponds in the continuous-time limit to a property of detailed balance on the Lindbladian $\call$ (see Alicki \cite{alicki_76} or the introduction of \cite{carlen_maas_17}). Under the condition that the baths $B_i$ are described by Gibbs states at temperatures $\beta_i$, we express the first and second principles of thermodynamics in terms of the mean energy fluxes: 
\begin{align*}
\sum_{i=1}^n J_i&=0 & \sum_{i=1}^n \beta_i J_i & \geq 0 ~.
\end{align*}
We show that for any list of fluxes $(J_1, \cdots, J_n)$ satisfying the above conditions (with a strict inequality) there exists a thermal model yielding theses energy fluxes (Proposition \ref{prop:thermal_machines}). The proof makes use of a design of quantum fridge borrowed from \cite{thermal_machine_small_maximal_efficiency_jpa11} and \cite{thermal_machine_small_prl10}. We apply this framework to the case of quasi-free fermionic systems, and prove that they satisfy a stronger inequality (Theorem \ref{theo:no_fridge}): provided $\beta_1\leq \beta_2\leq \cdots \leq \beta_n$ we have
\[
\sum_{i=1}^k J_i \leq 0
\]
for any $k\in \{1,\cdots, n\}$. 
In particular, there cannot be energy entering the bath of highest temperature. This theorem is inspired by the article of Eckmann and Zabey \cite{EckmannZabey04} on energy fluxes in classical harmonic networks. 

In the third section, we describe the large deviations of the random energy fluxes $N^i_t$ for quasi-free fermionic models. Using the results of \cite{JPW14}, we express the cumulant generating functional 
\[
e(\alpha)=\lim_{t\rightarrow +\infty} \frac{1}{t} \log \E\left(e^{\sum_{i=1}^n \alpha_i N_i}\right)
\]
in terms of the largest eigenvalue of a deformed Lindblad operator $\call_{\alpha}$. Note that if $L$ is the dimension of the one-particle space, the fermionic space is of dimension $2^L$, so  $\call_\alpha$ is of size $2^{2L}\times Z^{2L}$. We are able to reduce the computation of $e(\alpha)$ to the computation of the eigenvalues of an operator of dimension $4 L \times 4 L$, through the resolution of a Riccati equation (Theorem \ref{theo:e_alpha_fermionic}). This is formally similar to the study of the large deviations of entropy in classical harmonic networks \cite{JPS16}. We apply this to the numerical computation of the rate functional of large deviations for the fermionic chain, and show that the longest the chain is, the larger the fluctuations are.

\section{Mean energy exchanges for quantum Markovian unravelings }

We will use the following general notations: 

\begin{itemize}
\item Most Hilbert spaces considered will be finite dimensional. We write $\hh_S$ a Hilbert space associated with a system $S$, and $\Un_S$ the identity on this space; the space of operators on $\hh_S$ is written $\bb(\hh_S)$, and the set of states on $\hh_S$ is written $\gs(\hh_S)$. An operator on $\bb(\hh_S)$ is called a super-operator.  
\item For any state $\rho$ we write $S(\rho)=-\tr{\rho\log\rho}$ the Von Neumann entropy and $S(\rho|\sigma)=\tr{\rho(\log \rho-\log \sigma}$ the relative entropy with respect to a state $\sigma$.
\item For any state $\sigma>0$ we write $\Delta_\sigma(A)=\sigma A \sigma^{-1}$ the corresponding modular operator.
\item For two operators $A, B$ we write $[A, B]=AB-BA$ and $\set{A, B}=AB+BA$.
\end{itemize}

\subsection{The repeated interactions model in the continuous-time limit}

We consider a quantum system $S$ represented by a finite-dimensional Hilbert space $\hh_S$, in interaction with a bath $B$. In the Markovian approximation, the evolution on $\hh_S$ can be modeled by a quantum Markov semigroup.

\begin{defi}
A \emph{quantum Markov semigroup} (QMS) on a finite-dimensional Hilbert space $\hh_S$ is a family of linear maps on $\bb(\hh_S)$ which are completely positive and unity-preserving maps $(\Lambda^t)_{t\in I}$ where $I$ is $\N$ or $[0, +\infty)$, satisfying $\Lambda^{s+t}=\Lambda^s \Lambda^t$ for $s, t\in I$ such that $t\rightarrow \Lambda^t$ is continuous if $I=[0, +\infty)$.
\end{defi} 

Note that $\Lambda^t$ describes the evolution in the Heisenberg representation, the evolution in the Schr\"odinger representation being described by $(\Lambda^t)^*$.

In the case of discrete time $I=\N$ we can construct such a model by a repeated interaction process: the  bath is decomposed as a series of identical and independent sub-baths, interacting one after another with the system. The sub-bath model is a finite-dimensional Hilbert space $\hh_B$, and we fix a unitary $U$ on $\hh_S\otimes \hh_B$. Each bath subsystem is in the same state $\rho_B\in \bb(\hh_B)$, thus the evolution on the system is 
\[
(\Lambda^1)^*(\rho_S)=\tra_{\hh_B}\left(U(\rho_S\otimes \rho_B)U^*\right)~.
\]
In this article, we consider a QMS in continuous time $I=[0, T)$ obtained as the limit of discrete-time QMS when a parameter $\tau>0$ goes to zero, under suitable renormalization. We will take advantage of the easy interpretation of the repeated interaction model in discrete time to define quantities such as the entropy production and the energy exchanges between the bath and the system. The continuous-time limit of repeated interactions has been introduced by Attal in \cite{AttalToy}, and developed by Attal and Pautrat \cite{AttalPautrat}.

\begin{prop}[Adaptation of Theorem 22 of \cite{AttalPautrat}]
Fix some self-adjoint operators $H_S\in \bb(\hh_S)$ and $H_{SB}\in \bb(\hh_S\otimes \hh_B)$ and a state $\rho_B \in \gs(\hh_B)$ such that 
\begin{align}\label{eq:conditionV}
\tra_{B}(H_{SB}(\Un_S\otimes \rho_B))=0~.
\end{align}
For any time scale $\tau>0$ write $U_\tau$ the self-adjoint operator
\[
U_\tau=\exp(-i\tau H_S\otimes \Un_B -i\sqrt{\tau} H_{SB})
\]
and define the map $\Lambda_\tau$ by $(\Lambda_\tau)^*(\rho)=\tra_{B}(U_\tau (\rho\otimes \rho_B)U_\tau^*)$. Then for any $t\in [0, +\infty)$ the map $\Lambda^\ent{t/\tau}_\tau$ converges to a map $\Lambda^t$ as $\tau \rightarrow 0$ on the trace-norm topology over trace-class operators, locally uniformly in $t$. The family $(\Lambda^t)_{t\geq 0}$ is a continuous QMS with generator
\begin{align}\label{eq:lindblad}
\call(A)=i[H_S, A]+\Phi(A)-\frac{1}{2}\set{\Phi(\Un_S), A}
\end{align}
where $\Phi$ is the completely positive map defined by $\Phi^*(\rho)=\tra_{B}\left(H_{SB}(\rho\otimes \rho_B) H_{SB}\right)$.

Moreover, any continuous QMS can be obtained that way. We will call the triple $(H_S, H_{SB}, \rho_B)$ a repeated interaction model for the QMS $\Lambda$.
\end{prop}
This theorem was proved by Attal and Pautrat in \cite{AttalPautrat} in the case where $\rho_B$ is a pure state. We introduced this generalization with the condition of Equation \eqref{eq:conditionV} in \cite{AndreysFermions1}.

\begin{proof}
The idea to derive the continuous-time limit is to use the formula 
\[
\lim_{n\rightarrow +\infty} \left(I+\frac{A}{n}+o\left(\frac{1}{n}\right)\right)^n=e^A
\]
applied to $n=\ent{t/\tau}$, and $A=t\call$. 

Developing $U_\tau$ we obtain
\[
U_\tau=\Un_{SB}-i\sqrt{\tau} H_{SB}-i\tau H_S\otimes \Un_B-\frac{\tau}{2}  H_{SB}^2+O(\tau^{3/2})
\]
where $O(\tau^{3/2})=\tau^{3/2} R(\tau)$ for an operator $R(\tau)$ which is uniformly bounded as $\tau\rightarrow 0$. Thus, we have
\[
(\Lambda_\tau)^*(\rho)=\rho-i\sqrt{\tau} \tra_B\left(H_{SB}(\rho\otimes \rho_B)-(\rho\otimes \rho_B)H_{SB} \right)+\tau \call(\rho)+O(\tau^{3/2})
\]
where $\call$ is defined by \ref{eq:lindblad}. The term in $\sqrt{\tau}$ is zero because of Equation \eqref{eq:conditionV}. Thus, 
\[
(\Lambda_\tau)^{\ent{t/\tau}}=(\Un_{\bb(\hh_S)} +\tau \call+o(\tau) )^{\ent{t/\tau}}
\]
and $\tau=t/\ent{t/\tau}+o(\tau)$, so we obtain the convergence part of the Theorem.

Let us prove that any norm-continuous QMS $(\Lambda^t)_{t\geq 0}$ can be obtained by a repeated interaction model. By a theorem of Lindblad \cite{Lindblad1976} $(\Lambda^t)_{t\geq 0}$ admits a generator $\call$ which is of the form given by Equation \eqref{eq:lindblad} for some completely positive map $\Phi$. Since $\Phi$ is completely positive and norm-continuous we can write $\Phi(\rho)=\sum_{i=1}^{+\infty} L_i^* \rho L_i$ where the $L_i$ are bounded operators and $\sum_{i=1}^{+\infty} L_i^* L_i$ is bounded. To model $\Lambda^t$ as the limit of a repeated interaction model, we choose $\hh_B=l^2(\N)$ and $\rho=\ket{0}\bra{0}$ and 
\[
H_{SB}=\sum_{i=1}^{+\infty} L_i \otimes \ket{i}\bra{0}+L_i^*\otimes \ket{0}\bra{i}~.
\]
\end{proof}
\begin{rem}
Note that the operator $H_{SB}$ constructed in the last part of the proof satisfies the stronger condition 
\begin{align}\label{eq:strongConditionV}
\tr{H_{SB} (\Un_S\otimes f(\rho_B)}=0
\end{align}
for any function $f: \R\rightarrow \R$ (or equivalently, $\bra{i}H_{SB}\ket{i}=0$ for a Hilbert basis $\ket{i}$ in which $\rho_B$ is diagonal). This condition will be useful later.
\end{rem}

\subsection{Detailed balance} \label{subseq:energy_conservation}

We will consider only a special case of QMS, arising from a bath which is composed of several parts, each of them at thermal equilibrium with respect to a globally conserved pseudo-energy.

\begin{defi}
We call a \emph{thermal repeated interaction model} a repeated interaction model $(H_S,H_{SB},\rho_B)$ with a bath decomposed as $\hh_B=\bigotimes_{i=1}^n \hh_{B_i}$ in the state $\rho_B=\bigotimes_{i=1}^n \rho_{B_i}$, where the $\rho_{B_i}$ are Gibbs states:
\[
\rho_{B_i}=\frac{e^{-\beta_i K_{B_i}}}{\tr{e^{-\beta_i K_{B_i}}}}
\]
for some inverse temperatures $\beta_i \in \R$ and some self-adjoint operators $K_{B_i}\in \bb_{sa}(\hh_{B_i})$, with the following assumptions: $H_{SB}$ can be decomposed as $\sum_{i=1}^n H_{S B_i}$ with $H_{SB_i}$ acting only on $\hh_S\otimes \hh_{B_i}$, and there exists a self adjoint operator $K_S$ on $\hh_S$ with
\begin{align}\label{eq:invariantN}
[H_S, K_S]&=0 & \left[H_{SB_i},~ K_S\otimes \Un_B + \Un_S\otimes K_{B_i}\right]&=0~.
\end{align}
We call the operator $K_S$ the \emph{pseudo-energy} of the model.
\end{defi}

We choose the name pseudo-energy for $K_S$ because it is invariant and has the same dimension as the energy $K_{B_i}$, but it does not generate the dynamic.

The QMS arising from a thermal repeated interaction model are characterized by a detailed balance condition, defined as follow: 

\begin{defi}
A continuous-time QMS is said to satisfy the detailed balance condition with respect to a state $\sigma$ if its generator $\call$ can be written 
\[
\call(A)=i[H_S, A]+\Phi(A)-\frac{1}{2}\set{\Phi(\Un), A}
\]
where $H_S$ is a self-adjoint operator commuting with $\sigma$ and $\Phi$ is a completely positive map satisfying
\begin{align}\label{eq:db}
\Phi^*( A\sigma)\sigma^{-1}=\Phi(A)
\end{align}
for any operator $A$. 
\end{defi}
This condition is related to weaker conditions such as the time-reversal invariance (see \cite{JPW14}). Note that Equation \eqref{eq:db} means that $\Phi$ is self-adjoint with respect to the scalar product $\scal{A, B}_{\sigma, 0}=\tr{\sigma A^* B}$. It implies that $\Phi$ commutes with the modular operator $\Delta_\sigma$, and that for any $s$ it is self-adjoint with the scalar product $\scal{A, B}_{\sigma, s}=\tr{\sigma^{1-s}A^*\sigma^s B}$. See for example Carlen and Maas \cite{carlen_maas_17} or the original article of Alicki \cite{alicki_76}.
\newline

Alicki proved the following characterization of the strong detailed balance:

\begin{theo}\label{theo:alickidb}
Let $\beta\in \R$ and consider the Gibbs state $\sigma=\exp{-\beta K_S}/Z$. Let $\Lambda$ be a continuous-time QMS. Then $\Lambda^t$ satisfies the detailed balance for all $t$ if and only if it can be written in the Lindblad form \ref{eq:lindblad} with operator $H_S$ commuting with $K_S$, and completely positive map $\Phi$ of the form

\begin{align}
\Phi(A)&=\sum_{\delta\in sp([K_S, \bullet]), \delta \geq 0}~\sum_{i=1}^{n_{\delta}} e^{-\frac{\beta}{2}\delta~}L_{(\delta, i)}^* A L_{(\delta, i)}+e^{\frac{\beta}{2}\delta~}L_{(\delta, i)} A L_{(\delta,i)}^*
\end{align}
for some integers $n_\delta$, where the $L_{(\delta,i)}$ are operators satisfying $[K_S, L_{(\delta,i)}]=\delta L_{(\delta, i)}$ and $\tr{L_{(\delta,i)}}=0$.
\end{theo}

This theorem allows to make the link between the existence of a thermal model and the strong detailed balance condition, as follows.

\begin{prop}\label{prop:thermal=db}
A continuous-time QMS $\Lambda$ satisfies the strong detailed balance condition with respect to $\sigma=e^{-\beta K_S}/Z$ if and only if it admits a thermal repeated interaction model $(H_S,H_{SB}, \rho_B)$ with only one bath at inverse temperature $\beta$ and energy operator $K_S$. The model may be assumed to satisfy Condition \ref{eq:conditionV}.
\end{prop}
 
 \begin{proof}
To prove the sufficiency of our condition, assume that $\Lambda$ admits a thermal repeated interaction model with one bath $\hh_B$ at inverse temperature $\beta$, with energy operator $K_B$. Then $H_S$ commutes with $K_S$ so it commutes with $\sigma$, and using the formula 
\[
\Phi(A)=\tra_B\left(H_{SB} (A\otimes \Un_B) H_{SB} (\Un_S\otimes \rho_B)\right)~
\]
we obtain that for any operators $A, B\in \bb(\hh_S)$ we have
\begin{align*}
\scal{A, \Phi(B)}_{\sigma, 0}&=\tr{(\sigma\otimes \Un_B) (A^*\otimes \Un_B) H_{SB} (B\otimes \Un_B) H_{SB} (\Un_S\otimes \rho_B)} \\
\end{align*}
and since $[H_{SB}, \sigma\otimes \rho_B]=\left[H_{SB}, \frac{e^{-\beta(K_S+K_B)}}{Z_B}\right]=0$ this gives
\begin{align*}
\scal{A, \Phi(B)}_{\sigma, 0}&=\tr{(\sigma\otimes \rho_B)H_{SB}(A^*\otimes \Un_B)H_{SB} (B\otimes \Un_B)}\\
&=\scal{\Phi(A), B}_{\sigma, 0}~.
\end{align*}
Thus $\Phi$ is self-adjoint for the scalar product $\scal{A, B}_{\sigma, 0}$ so the detailed balance condition is satisfied.
\newline

To prove the necessary condition, we apply Theorem \ref{theo:alickidb}, obtaining an operator $H_S$ and operators $L_{(\delta,i)}$~. To construct $\hh_B$, we consider the set 
\[
\dd_+(K_S)=\bigcup_{\delta\in sp([K_S, \bullet])} \bigcup_{k=1}^{n\delta} \set{(\delta, k)}
\]
 and we take $\hh_B=\C^{\dd_+(K_S)}\otimes \C^2$, with Hilbert basis $\ket{\delta,k}\otimes\ket{i}$ for $(\delta, k)\in \dd_+(K_S)$ and $i\in \{\minus,\plus\}$. We consider the pseudo-energy operator

\[
K_B=\frac{1}{2} \sum_{(\delta, k)\in \dd_+(K_S)} \delta \ket{\delta, k}\otimes \Big(\ket{\plus}\bra{\plus}-\ket{\minus}\bra{\minus}\Big)~
\]
and take $\rho_B=e^{-\beta K_B}/Z$ with $Z=\tr{e^{-\beta K_B}}$. We define the interaction operator $H_{SB}$ by
\[
H_{SB}=Z\sum_{(\delta, k)\in \dd_+(K_S)}L_{(\delta,k)}\otimes \ket{\delta, k}\bra{\delta, k}\otimes \ket{\minus}\bra{\plus}+L_{(\delta,k)}^*\otimes \ket{\delta, k}\bra{\delta, k}\otimes \ket{\plus}\bra{\minus}
\]
The relation $[K_S, L_{(\delta, k)}]=\delta L_{(\delta, k)}$ implies the conservation of pseudo-energy
\[
\left[H_{SB}, K_S\otimes \Un_B+\Un_S\otimes K_B\right]=0~.
\]
It is easily checked that
\[
\tra_B\left(H_{SB}(A\otimes \Un_B)H_{SB} (\Un_S\otimes \rho_B)\right)=\sum_{(\delta,k)\in \dd_+(K_S)}  e^{-\frac{\beta}{2}\delta}L_{(\delta, k)}^*\rho L_{(\delta, k)}+e^{\frac{\beta}{2}\delta}L_{(\delta, k)} \rho L_{(\delta, k)}^*~.
\]
 \end{proof}

\begin{rem}
The detailed balance condition implies the time-reversal invariance: there exists an anti-linear involution of algebras $\Theta$ on $\bb(\hh_S)$ (called a time-reversal) with $\Theta(\sigma)=\sigma$ such that $\Theta(\call^*(\Theta(A)\sigma)\sigma^{-1})=\call$. This condition is strictly weaker than the detailed balance, indeed it is satisfied if and only if $\Phi$ can be written as
\[
\sum_{\delta\in sp([K_S, \bullet]), \delta\geq 0} \sum_{i=1}^{n_\delta} e^{-\frac{\beta}{2}\delta~}L_{(\delta, i)}^* A L_{(\delta, i)}+e^{\frac{\beta}{2}\delta~}M_{(\delta, i)} A M_{(\delta,i)}^*
\] 
where $L_{(\delta, i)}$ and $M_{(\delta, i)}$ both satisfy the same conditions as the $L_{(\delta,i)}$'s of Theorem \ref{theo:alickidb} and are related by $M_{(\delta, i)}=\Theta\left(L_{(\delta,i)}\right)$, which includes cases where $L_{(\delta,i)}\neq M_{(\delta,i)}$ and the detailed balance condition is not satisfied. 
\end{rem}

\begin{rem}
It is not clear how to define the detailed balance for a discrete-time Quantum Markov Semigroup $\Lambda^n$, since it is not possible to separate the generator in a unitary and a dissipative part. The best way to define it is probably the existence of a thermal repeated interaction model (with only one bath).
\end{rem}

%%%%%%%%%%%%%%%%%%%%%%%%%%%%%%%%%%%
%%%%%%%%%%%%%%%%%%%%%%%%%%%%%%%%%%%
%%%%%%%%%%%%%%%%%%%%%%%%%%%%%%%%%%%
%%%%%%%%%%%%%%%%%%%%%%%%%%%%%%%%%%%

%%%%%%%%%%%%%%%%%%%%%%%%%%%%%%%%%%%
%%%%%%%%%%%%%%%%%%%%%%%%%%%%%%%%%%%
%%%%%%%%%%%%%%%%%%%%%%%%%%%%%%%%%%%

\subsection{Energy and entropy fluxes}\label{subseq:energy_fluxes}

We now consider a thermal repeated interaction model $(H_S, H_{SB})$ with pseudo-energy operator $K_S$ and $n$ baths with inverse temperatures $\beta_1\leq \beta_2 \leq \cdots \leq \beta_n$ and energy operators $K_{B_i}$. We write $K_{tot}=K_S\otimes \Un_B+ \sum_{i=1}^n \Un_S\otimes K_{B_i}$ the global preserved pseudo-energy operator, define 
\[
\Phi_i(A)=\tr{H_{SB_i} (A\otimes \Un_B) H_{SB_i} \Un_S\otimes \rho_{B_i}}
\]
and write
\[
\call_i(A)=\Phi_i(A)-\frac{1}{2}\set{\Phi_i(\Un_S), A}
\]
so that each $\call_i$ generates a QMS with the strong detailed balance with respect to $\sigma_i=e^{-\beta_i K_S}/Z_i$ and $\call=-i[H_S, \bullet]+\sum_i \call_i$.
\newline

 In discrete times, we can define the energy increase of the $i$-th bath during one interaction as
\[
D_{\tau, i}(\rho)=\tr{(\Un_S\otimes K_{B_i}) \left(U_\tau \rho\otimes \rho_B U_\tau^*-\rho\otimes \rho_B\right)}
\]
and the total energy increase at time $\tau k$ as 
\[
D_{\tau, i, 0\rightarrow k}(\rho)=\sum_{l=0}^{k-1} D_{\tau, i}(\rho(l))
\]
where $\rho(l)=(\Lambda_\tau^*)^l(\rho)$. Let us study the limit as $\tau\rightarrow 0$ of $D_{\tau, i, 0\rightarrow \ent{t/\tau}}(\rho)$ for fixed $t\geq 0$. We have
\begin{align}\label{eq:estimate_Dtaui}
D_{\tau, i}(\rho)&=i\sqrt{\tau}\tr{[H_{SB_i}, K_{B_i}]\rho\otimes \rho_{B_i}}-\tau \tr{\call_i(K_S)\rho} +O(\tau^{\frac{3}{2}}) ~.
\end{align}
The term in $\tau$ is obtained by using the fact that $K_S\otimes \Un_B+\Un_S\otimes K_{B_i}$ commutes with $H_{SB_i}$. Thus, for the quantity $D_{\tau, i, 0\rightarrow \ent{t/\tau}}(\rho)$ to have a limit, we will impose the following assumption: 

\begin{assumption}\label{ass:strongConditionV}
For any bath index $i\in \{1, \cdots, n\}$ there is a basis $(\ket{j})_{j\in \{1, \cdots, \dim (\hh_{B_i})\}}$ of $\hh_{B_i}$ in which $K_{B_i}$ is diagonal and $\bra{j} H_{SB_i} \ket{j}=0$ for all $j$. Equivalently for any $\alpha\in \R$ we have
\[
\tra_{B_i}\left( H_{SB_i} \Un_{\hh_S}\otimes \rho_{B_i}^\alpha\right)=0~.
\]
\end{assumption}
This assumption is satisfied by the interaction $H_{SB}$ constructed in Proposition \ref{prop:thermal=db}. It is used to ensure the convergence of the energy fluxes, as in the following lemma.

\begin{lem}
Under Assumption \ref{ass:strongConditionV} the quantity $D_{\tau, i,  \ent{t/\tau}}(\rho)$ converges for all $t$ and $\rho$ as $\tau\rightarrow 0$ to a limit
\[
D_{i, t}(\rho)=\int_0^t J_{i}(\rho(s))ds
\]
where $\rho(s)=\Lambda^s(\rho)$ and $J_i$ is the flux of energy entering the $i$-th bath, equal to 
\begin{align}\label{eq:def_J}
J_i(\rho)=-\tr{\rho \call_i(K_S)}~.
\end{align}
\end{lem}

\begin{proof}
 Assumption \ref{ass:strongConditionV} implies that for any state $\rho$ the quantity $\tr{[H_{SB_i}, K_{B_i}]\rho\otimes \rho_{B_i}}$ is equal to zero, so Equation \eqref{eq:estimate_Dtaui} has no term in $\sqrt{\tau}$, so
 \[
 D_{\tau, i, 0\rightarrow \ent{t/\tau}}(\rho)=\tau \sum_{l=0}^{\ent{t/\tau}-1} J_i(\Lambda^{l\tau}(\rho)) +O\left(\tau^{3/2}\right)
 \] 
 converges, and has for differential $J_{i}(\rho(t))$.
 \end{proof}

Note that $\sum_{i=1}^n J_i(\rho)=-\tr{\rho\call(K_S)}$ so at equilibrium ($\call^*(\rho)=0$) the fluxes satisfies the first law of thermodynamics: $\sum_{i=1}^n J_i(\rho)=0$. As we shall see, it also satisfies the second law. 

\begin{prop}\label{prop:second_principle}
Under assumption \ref{ass:strongConditionV}, the fluxes satisfies
\begin{align*}
S(\rho(t))-S(\rho)+\sum_{i=0}^n \beta_i D_{i, t}(\rho)\geq 0~.
\end{align*}
\end{prop}

\begin{proof}
Let us consider the discrete-time model. The total entropy before the interaction is
\begin{align*}
S(\rho\otimes \rho_B)=S(\rho)+\sum_{i=1}^n \beta_i \tr{\rho_{B_i}  K_{B_i}}-\beta_i \log Z_i~.
\end{align*}

Now, let us consider the complete state after the interaction $\rho_{tot}'=U_\tau (\rho\otimes \rho_B) U_\tau^*$ and the partial state $\rho'=\tra_B(\rho_{\tau}')=\Lambda_\tau(\rho)$. Since the entropy is preserved by unitary evolution, $S(\rho_{tot}')=S(\rho\otimes \rho_B)$. Now we use a trick which goes back to \cite{pusz_passive_1978} (and has been used many times, see Section III of \cite{JP_landauer_2014}) : the relative entropy $S(\rho_{tot}'|\rho'\otimes \rho_B)$ is always positive, and it is equal to 
\begin{align*}
S(\rho_{tot}'|\rho'\otimes \rho_B)&=\tr{\rho_{tot}'\left(\log(\rho_{tot}')-\log(\rho')\otimes \Un_B-\Un_S\otimes \log(\rho_B)\right)}\\
&=-S(\rho_{tot}')+S(\rho')+\beta_i\tr{\rho_{tot}'\Un_S\otimes K_{B_i}}~.
\end{align*}
Thus we have
\[
S(\rho')-S(\rho)+\sum_{i=1}^n \beta_i D_{\tau,i}(\rho)=S(\rho_{tot}'|\rho'\otimes \rho_B) \geq 0~.
\]
Summing over all $k\in \{0, \ent{t/\tau}\}$ and taking the limit as $\tau\rightarrow 0$ allows to conclude.
\end{proof}

{\bf Remark: } The quantities $\call_i$, $J_i$ and $S(\rho)$ can all be defined without the help of the repeated interaction model, and also originate from other models such as the weak coupling limit. However, the repeated interaction model has the advantage of being easy to interpret; moreover some results on the repeated interaction models pass to the limit, as in the proof of Proposition \ref{prop:second_principle} above.

\subsection{Energy fluxes for stationary states and thermal machines}

The system being of finite dimension, there exists a trace-preserving projection (not necessarily an orthogonal one) $E$ on $\ker(\call)$ such that
\[
\lim_{T\rightarrow \infty} \frac{1}{T} \int_{0}^T e^{t\call^*}(\rho) dt=E(\rho)~.
\]
Here are some properties that the QMS may enjoy:
\begin{enumerate}
\item The system is ergodic if there is a unique state $\rho_\infty$ with $\call(\rho_\infty)=0$. Then $E(\rho)=\rho_\infty$ for any initial state $\rho$. 
\item The system is primitive if $0$ is a simple eigenvalue and is the only eigenvalue with real part zero. Then $E(\rho)=\rho_\infty=\lim_{t\rightarrow \infty} e^{t\call^*}(\rho)$ for any initial state $\rho$.
\item We will be interested in systems which are positivity improving, that is, primitive with a stationary state $\rho_\infty>0$. Equivalently (in finite dimension) for any nonzero operator $A$ with $A\geq 0$, we have $e^{t\call}(A)>0$ for all $t>0$.
\end{enumerate}

Let us consider the asymptotic of the mean energy flux 
\[
J_i=\lim_{t\rightarrow \infty}\frac{1}{T}D_{i, T}(\rho)=J_i(E(\rho))~.
\]
By Proposition \ref{prop:second_principle} they satisfy the first and second law of thermodynamics
\begin{align}
\sum_{i=1}^n J_i=0 \label{eq:firstlaw}\\
\sum_{i=1}^n \beta_i J_i\geq 0~.\label{eq:secondlaw}
\end{align}
We call \enquote{entropy production} the quantity $\sum_{i=1}^n \beta_i J_i$.

As we shall see, any list $(J_i)_{1\leq i \leq n}$ satisfying these conditions can be attained for a specific model, if the inequality in the second law is strict:

\begin{prop}\label{prop:thermal_machines}
Let $J_1, \cdots, J_n\in \R$ such that $\sum_{i=1}^n J_i=0$ and $\sum_{i=1}^n \beta_i J_i>0$. Then there exists a thermal repeated interaction model $(H_S, H_{SB})$ with $n$ baths, for which the complete Lindbladian $\call$ is positivity improving and such that $J_i(\rho_\infty)=J_i$ for all $i$ (where $\rho_\infty$ is the unique stationary state of the QMS).
\end{prop}

For two baths,  the conditions implies is $J_2=-J_1\geq 0$ if $\beta_1\leq \beta_2$, that is the energy flows from the hottest bath to the coldest one. The first nontrivial case arise with three baths, where we can have $J_3<0$ : the energy flowing from the hottest bath to the mild bath allows to pump energy from the coldest bath, as in a camping fridge (in which the hot bath is a gas stove, the mild bath is the ambient air and the cold bath is the inner of the fridge).

Later in the article we describe the class of quasi-free fermionic semigroups, for which this is not true, and only some very specific fluxes can be obtained. 
\vspace{0.5cm}

In order to prove this proposition, we introduce some special cases of thermal repeated interactions models.

{\bf The generalized depolarizing channel on a qubit: } This is the simplest non-trivial example of a thermal QMS. For any state $\sigma$ on $\hh_S$ the corresponding depolarizing channel has for generator
\[
\call_{\sigma, \lambda}^*(\rho)=\lambda(\sigma\tr{\rho}-\rho)
\]
for some positive real number $\lambda$ called the rate of the depolarizing channel. When $\sigma$ is faithful then $\call_{\sigma, \lambda}$ satisfies the detailed balance with respect to $\sigma$; indeed, it admits the following thermal repeated interaction model: take $\hh_B\simeq \hh_S$ and let $\rho_B=\sigma$ and let $H_{SB}$ be proportional to the swap operator: 
\[
H_{SB}=\sqrt{\lambda}\sum_{i,j=1}^{d_S} \ket{i}\bra{j}\otimes\ket{j}\bra{i}
\]
where $(\ket{i})_{i=1}^{d_S}$ is a Hilbert basis of $\hh_S$. 

Moreover, the semigroup is positivity improving when $\sigma$ is faithful and it is ergodic in general. The semigroup can be described explicitly: 
\[
\Lambda^t(\rho)=\left(\rho-\sigma\right)e^{-\lambda t}+\sigma~.
\]
\vspace{0.5cm}

{\bf Two depolarizing channels at different temperature:}

Let us consider two states $\sigma_1=e^{-\beta_1 K_S}/Z_1$ and $\sigma_2=e^{-\beta_2 K_S}/Z_2$, with $\beta_1\leq \beta_2$. Then we can combine two depolarizing channels corresponding to these states: 
\[
\call^*(\rho)=\lambda_1(\sigma_1\tr{\rho}-\rho)+\lambda_2(\sigma_2\tr{\rho}-\rho)
\]
for some positive real numbers $\lambda_1$ and $\lambda_2$. It is actually the depolarizing channel of rate $\lambda_1+\lambda_2$ with respect to the state
\[
\rho_\infty=\frac{\lambda_1 \sigma_1+\lambda_2\sigma_2}{\lambda_1+\lambda_2}~.
\]
Note that $\rho_\infty$ commutes with $K_S$ but it is not a Gibbs state with respect to the $K_S$ except in trivial cases. Associating one bath to each channel, we have
\[
J_1(\rho_\infty)=\frac{\lambda_1\lambda_2}{\lambda_1+\lambda_2} \tr{K_S (\sigma_1-\sigma_2)}~.
\]
Whenever $\beta_1<\beta_2$ it is possible to choose any negative value for $J_1(\rho_\infty)$ by tuning the rates $\lambda_1$ and $\lambda_2$, hence proving the proposition in the case of two baths. 
\vspace{0.5cm}

{\bf The quantum fridge:} This example was introduced by Linden, Popescu and Skrzypczyp in \cite{thermal_machine_small_prl10} (see also \cite{thermal_machine_small_maximal_efficiency_jpa11} where the solution is more detailed, and \cite{cooling_entanglement_bhlp} for more developments). It is a simple model of quantum fridge, where the energy of the coolest bath is pumped out by the use of two other baths. The solution can be explicitly computed but the description is more involved, and we refer to \cite{thermal_machine_small_maximal_efficiency_jpa11} for a complete discussion. Let us just describe the setup: the system is composed of three qubits: $\hh_S=(\C^2)^{\otimes 3}$. Write $P_1=\ket{1}\bra{1}\otimes \Un_{\C^2\otimes \C^2}$ the projector on the state $\ket{1}$ on the first qubit, and define $P_2$, $P_3$ similarly. We chose a pseudo-energy $K_S$ which acts independently on the qubits; 
\[
K_S=E_1 P_1+E_2 P_2+E_3~ P_3~.
\]
The energies $E_1, E_2, E_3$ are supposed nonzero.
We consider three baths with inverse temperatures $\beta_1>\beta_2>\beta_3$; the Hamiltonian $H_S$ on the system is defined by
\[
H_S=h K_S+g(\ket{010}\bra{101}+\ket{101}\bra{010})~.
\]
We assume $E_1+E_3=E_2$ so that $[H_S, K_S]=0$. This way, energy can flow from the hot bath to the middle bath only if some energy is pumped out of the cold bath. We take each bath acting on one qubit with the depolarizing channel corresponding to the Gibbs state $\sigma_i$ on this qubit at inverse temperature $\beta_i$, that is
\[
\sigma_i=\frac{1}{1+e^{-\beta_i E_i}}\left(\ket{0}\bra{0}+e^{-\beta_i E_i}\ket{1}\bra{1}\right)
\]
and 
\[
\call=i[H_S, \bullet]+\lambda_1 \call_1\otimes \Un_{\bb(\C^2\otimes \C^2)}+\lambda_2 \Un_{\bb(\C^2)}\otimes \call_2\otimes \Un_{\bb(\C^2)}+\lambda_3\Un_{\bb(\C^2\otimes \C^2)}\otimes \call_3~
\]
where $\call_i^*(\rho)=\sigma_i-\rho$ for any state $\rho$ on $\C^2$.  Note that  $\call$ is positivity improving since each $\call_i$ is positivity improving on its respective qubit. In \cite{thermal_machine_small_maximal_efficiency_jpa11} the stationary state $\rho_\infty$ is explicitly described, and the following facts are observed: 
\begin{lem}
There exists a parameter $\alpha\in \R$ depending on the $E_i$, $\beta_i$, $\lambda_i$ and $g,h$ such that
\begin{align}\label{eq:prop_flux}
J_1(\rho_\infty)=&\alpha E_1 & J_2(\rho_\infty)&=-\alpha E_2 & J_3(\rho_\infty)&=\alpha E_3
\end{align}
 Moreover, if all the parameters $E_i, \lambda_i, g, f$ are nonzero and if 
 \[
 \sum_{i=1}^n \beta_i E_i\neq 0
 \]
 then $\alpha \neq 0$ and $\alpha$ has the same sign as $\sum_{i=1}^n \beta_i E_i$. 
\end{lem}

\begin{proof}[Elements of proof]
The proportionality relation \ref{eq:prop_flux} can be proved directly: consider the observable $Q=P_1+2P_2+P_3$. Then $[H, Q]=0$ which implies
\[
\tr{\rho_\infty\sum_{i=1}^3\call_i^*(Q)}=\tr{\rho_\infty\call^*(Q)}=\tr{\call(\rho_\infty)Q}=0
\]
but the left-hand side is equal to
\[
\frac{J_1(\rho_\infty)}{E_1}+2\frac{J_2(\rho_\infty)}{E_2}+\frac{J_3(\rho_\infty)}{E_3}=0~.
\]
Since the sum of the fluxes equals $0$ and $E_1+E_3=E_2$ this implies the proportionality relation. The sign of $\alpha$ is constrained by Proposition \ref{prop:second_principle}. To prove that $\alpha$ is nonzero whenever $(\beta_1-\beta_2)E_1\neq(\beta_2-\beta_3)E_3$  we need the explicit solution described in \cite{thermal_machine_small_maximal_efficiency_jpa11}. We may just observe that $\rho_\infty$ is equal to the thermal equilibrium state $\sigma_1\otimes\sigma_2\otimes\sigma_3$ if and only if $(\beta_1-\beta_2)E_1=(\beta_2-\beta_3)E_3$, since $[H, \sigma_1\otimes\sigma_2\otimes\sigma_3]$ is equal to a nonzero coefficient times $e^{-\beta_1 E_1-\beta_3 E_3}-e^{-\beta_2E_2}$.
\end{proof}

 Note that by multiplying the rates $\lambda_i, f, g$ by some positive number $\mu$ the stationary state does not change, so the energy fluxes $J_i(\rho_\infty)$ are all multiplied by $\mu$. Thus, any fluxes $J_1, J_2, J_3$ satisfying $J_1+J_2+J_3=0$ and $\beta_1J_1+\beta_2 J_2+\beta_3J_3>0$ and such that $J_1<0$ or $J_3>0$ can be attained by tuning the parameters of the model.
\vspace{0.5cm}

We can now conclude the proof of Proposition \ref{prop:thermal_machines}.
\begin{proof}[Proof of Proposition \ref{prop:thermal_machines}]
The idea is to combine systems: consider $\hh_{S_1}$ and $\hh_{S_2}$ two thermal models, each with $n$ baths at the same respective temperatures $\beta_1, \cdots, \beta_n$ and  Lindbladians $\call_1, \call_2$ and fluxes $J_k^1$ and $J_k^2$.
Then the system $\hh_{S_1}\otimes \hh_{S_2}$ with Lindbladian $\call_1\otimes \Un_{S_2}+\Un_{S_1}\otimes \call_2$ can also be considered as a thermal model with $n$ baths, and the energy fluxes are additive: $J_k=J^1_k+J^2_k$. The key of the proof is that quantum fridges allow to construct systems with arbitrary small entropy production.

Let us show the Proposition by induction on $n$. Consider a list of fluxes $J_i$ satisfying the conditions of Proposition \ref{prop:thermal_machines}. 

If $n=2$ we can always obtain the fluxes with two depolarizing channels acting on the same qubit as in the example above.

If $n\geq 3$, let $\eps=\sum_{i=1}^n \beta_i J_i>0$. Define some fluxes $\tilde{J}_i$ for $i=1, 2, 3$ by
\[
\left\{\begin{array}{ll}
\tilde{J}_1&=J_1 \\
\tilde{J}_2&=\frac{\beta_1-\beta_3}{\beta_3-\beta_2} J_1-\frac{\eps}{2(\beta_3-\beta_2)}\\
\tilde{J}_3&=\frac{\beta_2-\beta_1}{\beta_3-\beta_2} J_1+\frac{\eps}{2(\beta_3-\beta_2)}
\end{array}\right.
\]
Then $\sum_{i=1}^3\beta_i J_i=\eps/2>0$ and $\sum_{i=1}^3 J_i=0$. Thus, the fluxes $\tilde{J}_i$ can be obtained from a model of quantum fridge as above. Moreover, consider the list of fluxes $\hat{J}_i$ defined by
\[
\hat{J}_i=\left\{\begin{array}{cc}
J_i-\tilde{J}_i & \text{ if $1\leq i \leq 3$} \\
J_i & \text{ if $i>3$}.
\end{array}\right.
\]
Note that $J_1=0$. Then $\sum_{i=1}^n \beta_i \hat{J}_i=\eps-\eps/2=\eps/2>0$ so the list $(\hat{J}_2, \cdots, \hat{J}_n)$ also satisfies the conditions of Proposition \ref{prop:thermal_machines} so by induction it admits a thermal model. By combining it to the quantum fridge with fluxes $(\tilde{J}_i)_{1\leq i\leq 3}$ we obtain the desired model.
\end{proof}

As shown by this proof, complex and rich examples can be obtained by making generalized depolarizing channels interact with the help of a Hamiltonian. The rest of this article is focused on another family of thermal models, arising from non-interacting fermions.

%%%%%%%%%%%%%%%%%%%%%%%%%%%%%%%%%%%%%%%%%%%%%%%%%%%%%%%%
%%%%%%%%%%%%%%%%%%%%%%%%%%%%%%%%%%%%%%%%%%%%%%%%%%%%%%%%
%%%%%%%%%%%%%%%%%%%%%%%%%%%%%%%%%%%%%%%%%%%%%%%%%%%%%%%
%%%%%%%%%%%%%%%%%%%%%%%%%%%%%%%%%%%%%%%%%%%%%%%%%%%%%%%%

\subsection{Quasi-free fermionic semigroups arising from a repeated interaction model}

In this subsection we briefly introduce quasi-free fermionic systems and we describe a class of quantum semigroups on such systems. We then show that the asymptotic energy fluxes in thermal quasi-free fermionic systems are always trivial, in the sense that they can be decomposed as the sum of fluxes which stream between two of the baths, always from the hottest bath to the coldest one. Thus, it is never possible to pump energy from the coldest bath.

\subsubsection{Fermionic systems, quadratic Hamiltonians and quasi-free states}

In this part we recall the basic definitions and fix notations on fermionic systems. This is essentially a shorter version of the introduction of \cite{AndreysFermions1}. For a more general introduction to fermionic and bosonic spaces, see Derezi\'nski and G\'erard \cite{DerezinskiGerard}.
\vspace{0.5cm}

 Let us consider a (finite-dimensional) Hilbert space $\hh_0$, called the \enquote{one-particle space}, and let us fix a Hilbert basis $\ket{1},\cdots, \ket{L}$ of this space. The fermionic space constructed from $\hh_0$ is written $\hh=\Gamma(\hh_0)$. It is of dimension $2^L$, and has for orthonormal basis $\set{\ket{u_1, \cdots,u_L}~|~u_1, \cdots, u_L\in \{0, 1\}}$. We write $c_i$ the annihilation operator and $c_i^*$ the creation operator corresponding to the one-particle state $\ket{i}$, so that
\[
c_i\ket{u_1, \cdots, u_L}=(-1)^i \delta_{u_i=1} \ket{u_1, \cdots, u_i-1, \cdots, u_L}
\]
and the anticommutation relations are satisfied: 
\begin{align}
\set{c_i, c_j}&=\set{c_i^*, c_j^*}=0 \\
\set{c_i^*, c_j}&=\delta_{i,j} \Un~
\end{align}
where $\set{A, B}=AB+BA$. More generally for any vector $v\in \hh_0$ we have two operators $c_{v}^*=\sum_{i=1}^L v_i c_i^*$ and $c_{v}=\sum_{i=1}^L \overline{v_i} c_i$. 

We write $\gamma_1, \cdots, \gamma_{2L}$ the Majorana operators, defined by
\begin{align*}
\gamma_i&=c_i+c_i^* & \gamma_{i+L}&=-i (c_i -c_i^*)\\\
\end{align*}
for $i\leq L$. They are self-adjoint and satisfy the anticommutation relation 
\begin{align}
\set{\gamma_i, \gamma_j}=2\delta_{i,j}~.
\end{align}

We consider the Hilbert space $\yy=\hh_0\oplus \overline{\hh_0}$, where $\overline{\hh_0}$ is a Hilbert space endowed with an anti-unitary map $s$ to $\hh_0$. It is called the \emph{phase space} ; it has a Hilbert basis $e_1, \cdots, e_{2L}$ defined by $e_i=\ket{i}\oplus 0$ and $e_{i+L}=0\oplus s(\ket{i})$ for $i\leq L$. We write $\varphi$ the field operator, defined as a linear application from $\yy$ to $\bb(\hh)$ by
\begin{align*}
\varphi(e_i)&=c_i^* & \varphi(e_{i+L})&=c_i~.
\end{align*}
The space $\yy$ is endowed with the anti-linear involution $\xi(x\oplus s(y))=y\oplus s(x)$, with the property $\varphi(\xi(z))=\varphi(z)^*$. The anticommutation relations writes:
\[
\{\varphi(x), \varphi(y)\}=\scal{\xi(x), y}~.
\]

Another interesting basis of $\yy$ is the orthogonal basis $f_1, \cdots, f_{2L}$ defined by $f_i=e_i+e_{i+L}$ and $f_{i+L}=-i(e_i-e_{i+_L})$, so that $\varphi(f_i)=\gamma_i$. In this basis, $\xi$ is just the componentwise complex conjugation. The basis $e_1, \cdots, e_{2L}$ will be called the creation/annihilation basis while the basis $f_1, \cdots, f_{2L}$ will be called the Majorana basis. 

 For any operator $M: \yy\rightarrow \yy$ we will write $M^T=\xi M^* \xi$. In the Majorana basis, it corresponds to the transposition, while in the creation annihilation basis, we have
 \begin{align*}
 \begin{pmatrix} A & B \\ C & D \end{pmatrix}^T=\begin{pmatrix} D^T & B^T \\ C^T & A^T \end{pmatrix}~.
 \end{align*}

{\bf Row and column operators:} Another useful way of seeing $\varphi$ is as a row of operators: 
\begin{defi}

The row operator is the operator $F^*: \hh\otimes \yy \rightarrow \hh$ defined by 
\[
F^*(\ket{u_1, \cdots, u_L}\otimes x)=\varphi(x)\ket{u_1, \cdots, u_L}~.
\]
Its adjoint is the operator $F: \hh\rightarrow \hh\otimes \yy$ with
\[
F(\ket{u_1, \cdots, u_L})=\sum_{i=1}^{2L} \Big(\varphi(e_i) \ket{u_1, \cdots, u_n}\Big)\otimes e_i~.
\]

\end{defi}

Expressed in the creation/annihilation basis of $\yy$, the operator $F$ forms a column of operators: 
\[
F_{c}^*=\begin{pmatrix} c_1 \\ \vdots \\ c_L \\ c_1^* \\ \vdots \\ c_L^* \end{pmatrix}~.
\]
In the Majorana basis its form is 
\[
F_{f}^*=\begin{pmatrix} \gamma_1 \\ \vdots \\ \gamma_{2L} \end{pmatrix}~.
\]

In what follows, we recall the definition of quadratic operators, Bogoliubov transform and quasi-free states.

\begin{defi}
A quadratic operator on $\hh=\Gamma(\hh_0)$ is an operator of the form
\[
A=\frac{1}{2}F^*\left(\Un_{\hh}\otimes T\right) F~
\]
for some operator $T$ on $\yy$. If $T^f$ is its matrix in the Majorana basis, we have
\[
A=\sum_{1\leq i, j \leq 2L} \frac{1}{2}[T^f]_{i,j} \gamma_i \gamma_j
\] 
Up to replacing $A$ by $A+\alpha \Un$ for some $\alpha\in \R$ we may always assume that $T^T=-T$. Under this condition, $A$ is self-adjoint if and only if $T^*=T$, equivalently $\xi T \xi=-T$, or equivalently there exists a real antisymmetric matrix $R$ such that $iR$ is the matrix of $T$ in the Majorana basis. 
\end{defi}

We will implicitly write $T$ for $\Un_{\hh}\otimes T$ when there is no possible confusion, and we will write $d\Gamma(T)=A=\frac{1}{2} F^* T F$. 

 A quadratic observable can also be expressed in the creation/annihilation basis:
\[
A=\sum_{1\leq i, j\leq 2L} T^{c}_{i, j} (c^\sharp_i)^* c^\sharp_j
\]
where for $i\leq L$ we define $c_i^\sharp=\varphi(e_i)=c_i$ and $c_{i+L}^\sharp=\varphi(e_{i+L})=c_i^*$, and $T^{c}$ is the matrix of $T$ in the creation/annihilation basis. We have $\xi T \xi=-T$ if and only if $T^c$ is of the form
\[
T^{c}=\begin{pmatrix}
A & B\\
-\overline{B} & -\overline{A}
\end{pmatrix}=\frac{1}{2} \begin{pmatrix}
1 & i \\
1 & -i
\end{pmatrix} T^{f}\begin{pmatrix}
1 & 1 \\
-i  & i 
\end{pmatrix}
\]
and under this condition $T$ is self-adjoint if and only if $A$ is self-adjoint and $B$ is antisymmetric (in the sense that $B^T=-B$).

The exponential of quadratic operators are characterized the following way: 

\begin{prop}\label{prop:commutation_formula_exp_F}
For any operator $M=\exp{T}$ on $\yy$ with $\xi T=-T\xi$ the operator $\Gamma(M)=\exp{\frac{1}{2} F^*T F}$ on $\hh$ satisfies
\[
\varphi(Mx)=\Gamma(M)\varphi(x)\Gamma(M)^{-1}~.
\]
In terms of the column operator, 
\begin{align}\label{eq:commutation_F_M}
(\Un_{\hh}\otimes M) F=(\Gamma(M)^{-1}\otimes \Un_{\yy}) F~ \Gamma(M)~.
\end{align}
\end{prop}
{\bf Remark: } There is a redundancy in the expression $F^*TF$ since the terms $\gamma_i\gamma_j$ and $\gamma_j\gamma_i$ can be regrouped, which is at the origin of the factor 2 in $\Gamma(M)=\exp{\frac{1}{2} F^*TF}$.

The Bogoliubov transforms are a very important class of unitary operators on the phase space.
\begin{defi}
A \emph{Bogoliubov transform} is a unitary operator $U: \yy\rightarrow \yy$ satisfying one of the equivalent conditions: 
\begin{enumerate}
\item We have $\xi U=U\xi$.
\item The matrix of $U$ in the Majorana basis is real.
\item The matrix of $U$ in the creation/annihilation basis is of the form
\[
U_c=\begin{pmatrix} \nu & \gamma \\ \overline{\gamma} & \overline{\nu}
\end{pmatrix}~.
\]
\item There exists a self-adjoint operator $T$ on $\yy$ such that $\xi T=-T\xi$ and $U=\exp{iT}$.
\item There exists a unitary operator $V$ on $\hh$ such that $\varphi(Ux)=V \varphi(x) V^*$ for any $x\in \yy$. 
\end{enumerate}
\end{defi}

The interest of Bogoliubov transforms is that the operators $\tilde{c_i}=\varphi(Ue_i)$ also satisfy the anticommutation relations. Moreover, Bogoliubov transforms are a generalization of unitary transformations on the one-particle space: if $V$ is any unitary on $\hh_0$, we can define the Bogoliubov transform $U=V\oplus \overline{V}$ on $\yy$, and the creation operators $\tilde{c}_i$ corresponding to the new basis $V\ket{1}, \cdots, V\ket{L}$ are simply the $\Gamma(U) c_i \Gamma(U)^*$. 
\vspace{0.5cm}

Let us now turn to the study of states on $\hh$. Much information on a state can be obtained by studying its covariance matrix: 
\begin{defi}
The \emph{covariance matrix} of a state $\rho \in \gs(\hh)$ is the matrix $\cov(\rho)$ on $\yy$ defined by
\[
\cov(\rho)=\tra_{\hh}\left(\rho F F^*\right)~.
\]
In the Majorana basis, 
\[
[\cov(\rho)_f]_{i, j}=\tr{\rho \gamma_i\gamma_j}~.
\]
\end{defi}

The covariance matrix is covariant under the evolution by a Bogoliubov transform, in the following sense: if $U$ is a Bogoliubov transform then 
\[
\cov(\Gamma(U) \rho \Gamma(U)^*)=U \cov(\rho) U^*~.
\]
Any covariance matrix is of the form
\[
\cov(\rho)=\frac{1}{2} \Un +M
\]
where $M$ is a self-adjoint operator with $\xi M \xi=-M$.

In the creation annihilation basis, for $i\leq L$ we have $\cov(\rho)_{i,i}=\tr{\rho c_i c_i^*}=1-\tr{\rho c_i^*c_i}$. The number $\tr{\rho c_i^* c_i}$ can be interpreted as the mean number of particles in the mode $i$. For this reason, some author prefer to define the covariance matrix as $\Un_\yy-\tra_{\hh}\left(\rho F F^*\right)$ (which is also the transpose of our definition). 
\vspace{0.5cm}

The quasi-free states form a class of states which are fully determined by their covariance matrix. 
\begin{defi}
A state $\rho$ is called a \emph{quasi-free state} if it satisfies the Wick formula: for any $i_1, \cdots, i_n \in \{1, \cdots, 2L\}$ we have
\begin{align}\label{eq:wick}
\tr{\rho~ \gamma_{i_1}\cdots \gamma_{i_n}}&=\left\{\begin{array}{cc}
\sum_{\sigma\in \pp_n} (-1)^{\eps(\sigma)} \prod_{l=1}^{n/2} ~\tr{\rho ~\gamma_{\sigma(2l)}\gamma_{\sigma(2l-1)}} & \text{if $n$ is even} \\
0& \text{if $n$ is odd}~
\end{array}\right.
\end{align}
where $\pp_n$ is the set of pairings of the set $\{1, \cdots, n\}$, that is the set of permutations $\sigma$ of $\set{1, \cdots, n}$ with $\sigma(2i) < \sigma(2i+1)$ for all $i\in \{1, \cdots, n/2\}$, and $\eps(\sigma)$ is the signature of the permutation $\sigma$.
\end{defi}

Any quasi-free state which is faithful is a Gibbs state for a quadratic Hamiltonian.

\begin{prop}
Any faithful state $\rho$ is quasi-free if and only if there exists a quadratic Hamiltonian $H=\frac{1}{2} F^* T F$ such that $\rho=e^{-\beta H}/\tr{e^{-\beta H}}$ for some $\beta\in \R$. The covariance matrix of $\rho$ is related to $T$ the following way: 
\begin{align}\label{eq:cov_T}
\cov\left(\frac{e^{-\frac{\beta}{2} F^* T F}}{\tr{e^{-\frac{\beta}{2} F^* T F}}}\right)=(\Un+e^{-\beta T})^{-1}~.
\end{align}
\end{prop}

Finally, we define the number operator: 
\begin{defi}
The \emph{number operator} is the operator on $\hh$ defined by
\begin{align}
N=\sum_{i=1}^L c_i^*c_i=\frac{1}{2} F^*\left( \Un_{\hh_0} \oplus (-\Un_{\overline{\hh_0}})\right) F+\frac{L}{2} \Un_{\hh_S}~.
\end{align}
 Any operator commuting with $N$ is called a gauge-invariant operator, and any operator commuting with $(-1)^N$ is called an even operator, the operators anti-commuting with $(-1)^N$ being called odd. We write $Even(\hh)$ the space of even operators and $Odd(\hh)$ the space of odd operators. 
 \end{defi}

\subsubsection{Quasi-free fermionic semigroups}

Let us consider a fermionic system $\hh_S=d\Gamma(\hh_{S, 0})$ with $\hh_{S,0}$ of finite dimension $L_S$, and some reference basis $\ket{1^S}, \cdots, \ket{L_S^S}$, for which the creation operators are written $c_{S, i}$ and Majorana operators $\gamma_{S, i}$. We also write $N_S$ the number operator on $S$, and $F_S$ the column operator and $\yy_S$ its phase space. A quasi-free fermionic semigroups on $\hh_S$ is a quantum semigroup whose Lindbladian is of the form
\[
\call(\rho)=-i [d\Gamma(T_S), \rho]+\sum_{1\leq i, j\leq L} A_{i,j} \left(\gamma_{S, i} \rho \gamma_{S,j}-\frac{1}{2}\set{\gamma_{S, j}\gamma_{S, i}, \rho}\right)~
\]
for some self-adjoint matrix $A_{i,j}$. We are specifically interested in quasi-free semigroups arising from a quasi-free repeated interaction model: we consider a \enquote{bath} system $\hh_B=d\Gamma(\hh_B)$, with $\hh_{B, 0}$ of finite dimension $L_B$ and some reference basis $\ket{1^B}, \cdots, \ket{L_B^B}$ and creation and Majorana operators $c_{B, i}$ and $\gamma_{B, i}$. We wish to make it interact with $\hh_S$ by the use of a quadratic Hamiltonian. For this, we need to see $\hh_S\otimes \hh_B$ as a fermionic system. There are many possible isomorphisms between $\hh_S\otimes \hh_B$ and $\hh_{SB}=\Gamma(\hh_{S, 0}\oplus \hh_{B, 0})$, the two standard ones being $\ket{u_S}\otimes \ket{u_B}\mapsto \ket{u_S}\wedge \ket{u_B}$ and $\ket{u_S}\otimes \ket{u_B}\mapsto \ket{u_B}\wedge \ket{u_S}$. The second isomorphism is the most convenient in our case. A basis of $\hh_{S, 0}\oplus \hh_{B, 0}$ is the basis
\[
0_S\oplus \ket{1^B}, \cdots, 0_S\oplus\ket{L_B^B},~ \ket{1^S}\oplus 0_B, \cdots, \ket{L_S^S}\oplus 0_B~.
\]
The corresponding creation operators $c_{0_S\oplus \ket{i^b}}$ and $c_{\ket{i^S}\oplus 0_B}$ on $\hh_SB$ are identified respectively with the operators $c_{S, i}\otimes (-1)^N_B$ and $\Un_{\hh}\otimes c_{B, i}$ on $\hh_S\otimes \hh_B$. We likewise consider $F_S$ and $F_B$ as column operators acting on $\hh_S\otimes \hh_B$, and identify $\hh_S\otimes \hh_B$ with  $\hh_{SB}$ in what follows.

The repeated interaction model we consider is the following: 
\begin{itemize}
\item The Hamiltonian on $\hh_S$ is 
\[
H_S=\frac{1}{2} F_S^* T_S F_S
\]
 for some self-adjoint $T_S$ on $\yy_S$ with $\xi T_S\xi=-T_S$. We write $iR_S$ its matrix in the Majorana basis, where $R_S$ is an antisymmetric $2L_S\times 2L_S$ matrix.
\item The interaction Hamiltonian is 
\[
H_{SB}= F_S^* \Theta F_B=\frac{1}{2} \left( F_S^* \Theta F_B+F_B^* \Theta^* F_S \right)
\]
where $\Theta: \yy_B\rightarrow \yy_S$ is a linear operator with $\xi \Theta \xi=-\Theta$. We write $iW$ its matrix in the Majorana basis, where $W$ is a $2L_S\times 2L_B$ matrix.
\item The state $\rho_B$ on $\hh_B$ is a quasi-free state. We write its covariance matrix $\cov(\rho_B)=M_B$, in the Majorana basis it is of the form $\frac{1}{2}\Un+i R_B$ where $R_B$ is a real antisymmetric matrix.
\end{itemize}

First we check that the condition are satisfied for the continuous time-limit to exist: 
\begin{lem}
Assumption \ref{ass:strongConditionV} is satisfied for $H_{SB}$ and $\rho$, that is: for any $\alpha \in \R$ we have
\[
\tra_{B}\left(\rho_B^\alpha H_{SB} \right)=0~.
\]
\end{lem}
\begin{proof}
The operator $\rho_B^\alpha$ is even and $H_{SB}\in Odd(\hh_S)\otimes Odd(\hh_S)$ so the operator $\rho_B^\alpha H_{SB}$ is in $Odd(\hh_S)\otimes Odd(\hh_B)$, and so its partial trace with respect to $\hh_B$ is zero.
\end{proof}

The continuous-time QMS $(\Lambda^t)_{0\leq t \leq \infty}$ constructed from this model has for generator $\call$ defined by
\[
\call(A)=i[H_S, \rho]+\frac{1}{2}\set{\Phi(\Un_S), A}+\Phi(A)
\]
for any $A\in \bb(\hh_S)$, where $\Phi(A)=\tra_{B}\left(\Un_S\otimes \rho_B)H_{SB} (A\otimes \Un_B) H_{SB}\right)$. The form of $\Phi$ can be made more explicit: 
\begin{align}\label{eq:phi}
\Phi(A)=F_S^* \left(A\otimes \Theta M_B \Theta^*\right) F_S=\sum_{1\leq i,j\leq L_S} [\Theta M_B \Theta^*]^f_{i,j} \gamma_{S, i} A \gamma_{S, j}~
\end{align}
where $[\Theta M_B \Theta^*]^f$ is the matrix of $\Theta M_B \Theta^*$ in the Majorana basis. 

Let us write $M_S(t)$ the covariance matrix of $\rho_S(t)$. Using Equation \eqref{eq:commutation_F_M} with $\Gamma(U(\tau))=\exp\left(-i\tau H_S-i\sqrt{\tau} H_{SB}\right)$ and passing to the limit as $\tau\rightarrow 0$ we obtain
\begin{align}\label{eq:cov_evolution}
\frac{d}{dt} M_S(t)=\left(-iT_S-\frac{1}{2} \Theta \Theta^*\right) M_S(t)+M_S(t)\left(iT_S-\frac{1}{2}\Theta \Theta^*\right)+\Theta M_B \Theta^*~.
\end{align}
Thus, knowing $M_t$ is sufficient to compute $M_s$ for $s\geq t$, even when we do not know anything else on $\rho_S(t)$. Moreover, if $\rho_S(0)$ is a quasi-free state, then $\rho_S(t)$ is a quasi-free state for all $t$, and we have the following: 

\begin{theo}\label{theo:kalman}
Consider the subspace $K(T_S, \Theta)$ of $\yy$ generated by the ranges of $T_S^k \Theta$ for $k\in\N$. Then 
the QMS is ergodic if and only if there is a unique solution $M_\infty$ to the Lyapunov equation
\[
\left(-iT_S-\frac{1}{2} \Theta \Theta^*\right) M_\infty+M_\infty\left(iT_S-\frac{1}{2}\Theta \Theta^*\right)+\Theta M_B \Theta^*=0~.
\]
 This is equivalent to the condition $K(T_S, \Theta)=\yy$. The equilibrium state is then the quasi-free state of covariance matrix $M_\infty$. If the state $\rho_B$ in the repeated interactions model of the QMS is faithful, then the stationary state $\rho_\infty$ of the QMS is also faithful, and the QMS is positivity improving.
\end{theo}

 This criterion is called the Kalman criterion (from the theory of control of linear systems) and $K(T_S, \Theta)$ is called the Kalman space. See \cite{AndreysFermions1} for a proof; a similar theorem was also proved by Prosen in \cite{Prosen2008} in the case where the QMS is restricted to even operators. This theorem is just algebraic when restricted on quasi-free states, the hard part being to treat the case where the initial state is not quasi-free, particularly when the state $\rho_\infty$ is not faithful.
 
 Note that the condition is independent of $M_B$. If it is satisfied, then the operator $G=-iT_S-\frac{1}{2} \Theta \Theta^*$ has all its eigenvalues with strictly negative real part, and the solution $M_\infty$ is given by
 \[
 M_\infty=\int_0^{+\infty} e^{t G} \Theta^* M_B \Theta e^{t G^*} dt~.
 \]

\subsubsection{Thermal quasi-free fermionic semigroups} \label{subsub:qf_semi}

We consider quasi-free fermionic semigroups as described above for which there is a conserved quadratic pseudo-energy: we take
 \[
 \hh_{B,0}=\bigoplus_{i=1}^n \hh_{B_i, 0}
\]
with dimensions $L_{B_i}$, phase space $\yy_{B_i}$ and $M_B=\bigoplus_{i=1}^n M_{B_i}$, such that $\Theta=\sum_{i=1}^n \Theta_i$ for some operators $\Theta_i: \yy_{B_i}\rightarrow \yy_S$ with $\xi_i \Theta_i\xi_i=-\Theta_i$. We fix some self-adjoint operators $\kappa_S$ on $\yy_S$ and $\kappa_{i}$ on $\yy_{B_i}$, all anti-commuting with the conjugation $\xi_i$, and we define
\begin{align}
K_S&=\frac{1}{2}F_S^*\kappa_S F_S   & K_{B_i}&=\frac{1}{2} F_{B_i}^* \kappa_{i} F_{B_i}~.
\end{align}
We assume that they are conserved by the dynamic generated by $H_S$ and $H_{SB}$, which is equivalent to
\begin{align}\label{eq:db_quadratic_operators}
[T_S, \kappa_S]&=0 & \Theta_i \kappa_{i}=\kappa_S \Theta_i~\text{for all i}.
\end{align}
Moreover, we assume that the $\rho_{B_i}$ are thermal with respect to the $K_{B_i}$, that is $\rho_{B_i}=e^{-\beta_i K_{B_i}}/Z_i$, or in terms of the covariance matrix, 
\begin{align}
M_{B_i}&=\left(\Un+e^{-\beta_i \kappa_{B_i}}\right)^{-1}~.
\end{align}

Let us write $M_\beta$ the covariance matrix of the Gibbs' state $e^{-\beta K_S}/Z$, that is
\begin{align}\label{eq:defi_mbeta}
M_{\beta}=\left(1+e^{-\beta \kappa_S}\right)^{-1}
\end{align}
Note that by the detailed balance condition \ref{eq:db_quadratic_operators} we have
\[
\Theta_i M_{B_i} \Theta_i^*=\Theta_i \Theta_i^* M_{\beta_i}~.
\]
We define
\begin{align}
D_i(A)=\frac{1}{2}\set{\Theta_i \Theta_i^*, A}
\end{align}
so that for any state $\rho$ of covariance matrix $M_S$ we have
\begin{align}
\tr{\call_i^*(\rho) F F^*}=D_i(M_{\beta_i}-M_S)~.
\end{align}

 The flux $J_i(\rho)$ entering the i-th reservoir is then
\begin{align}\label{eq:flux_fermions}
J_i(\rho)=\frac{1}{2}\tra_{\yy_S}\left(\kappa_S D_i(M_{\beta_i}-M_S\right))~.
\end{align}
Indeed, we have
\begin{align*}
J_i(\rho)&=-\frac{1}{2}\tr{\call_i^*(\rho) F^* \kappa_S F}\\
&=-\frac{1}{2}\sum_{1\leq k, l\leq L_S}[\kappa_S]^f_{k,l} \tr{\call_i^*(\rho)\gamma_k \gamma_l} \\
&=-\frac{1}{2}\sum_{1\leq k, l\leq L_S} [\kappa_S]^f_{k,l} [D_i(M_{\beta_i}-M_S)]^f_{k,l} \\
&=+ \frac{1}{2} \tra_{\C^{2L_S}}\left(([\kappa_S]^f)D_i(M_{\beta_i}-M_S)]^f\right)\\
\end{align*}
since the matrix $[\kappa_S]^f$ of $\kappa_S$ in the Majorana basis is of the form $i R$ with $R$ real antisymmetric.

By Proposition \ref{prop:second_principle}, if $\rho_\infty$ is a stationary state then $\sum_{i=1}^n \beta_i J_i(\sigma)\geq 0$. In the following theorem we show that there is a stronger constraint on the $J_i(\sigma)$, which prevent non-trivial thermal machines such as the quantum heat pump to be designed.

\begin{theo}\label{theo:no_fridge}
Consider a thermal quasi-free QMS as above. Let $\rho_\infty$ be a stationary state and $J_i=J_i(\rho_\infty)$. Then there exists a family of fluxes $(J_{i,j})_{1\leq i,j\leq n}$ with $J_{j,i}=-J_{i,j}$ and $J_{i,j}\geq 0$ if $\beta_i \geq \beta_j$ such that for any $i\leq n$,
\[
J_i=\sum_{j=1}^n J_{i,j}~.
\]
In other words, the fluxes $J_i$ can be obtained from a combination of systems, each of them involving only two of the bath.
\end{theo}

This theorem is inspired by Lemma 1 of Eckmann and Zabey \cite{EckmannZabey04}, in which they consider a system of oscillators coupled by springs and driven by Gaussian heat bath. They show something weaker that this theorem, namely that the bath of lower temperature cannot be pumped out of energy. Our proof is an elaboration of theirs, and also applies to the system they consider.

\begin{proof}
First, we reformulate the theorem with a majorization condition: 
\begin{lem}
Let us assume that $\beta_1\leq \beta_2\leq \cdots \leq \beta_n$. Then the fluxes $J_i$ can be decomposed as a sum of $J_{i,j}$ as above if and only if they satisfy that for all $k\leq n$, 
\[
\sum_{i=1}^k J_i\leq 0~.
\]
\end{lem}

This is easily proved by induction. We assume $\beta_1\leq \cdots \leq \beta_n$ in the rest of the proof.

\vspace{0.5cm}

We shall start by assuming that $\Lambda^t$ is positivity improving.  Then the eigenvalues of the operator $G=-iT_S-\frac{1}{2} \Theta \Theta^*=-iT_S-\frac{1}{2} \sum_{i=1}^n \Theta_i \Theta_i^*$ have strictly negative real part. We define the function $F$ on $\bb(\yy)$ which is the inverse of $M\rightarrow GM+MG^*$, that is
\[
F(M)=\int_0^\infty e^{tG}M e^{tG^*} dt~.
\]
Thus, $M_S=F\left(\sum_{i=1}^n D_i(M_{\beta_i})\right)$ is the solution of 
\[
G M_S+M_SG^*+\sum_{i=1}^n D_i(M_{\beta_i})=0
\]
so it is the covariance matrix of $\rho_\infty$. Note that $F$ is linear as a map on $\bb(\yy_S)$, and for any $\beta\in \R$ it satisfies 
\begin{align}\label{eq:zprope_F}
F(\sum_{i=1}^n D_i(M_\beta))=M_\beta~.
\end{align}

Let us fix a $k\leq n$ and show that
\[
\sum_{i=1}^k \tr{\kappa_S D_i(M_{\beta_i}-M_S)}\leq 0~.
\]
We define the following order relation between self-adjoint operators on $\yy$. 
\begin{defi}
Write $P_+$ the projection on the positive eigenspace of $\kappa_S$, and $P_-$ the projection on the negative eigenspace of $\kappa_S$. We say a self-adjoint matrix $M$ on $\yy$ is $\kappa_S$-positive (and we write $0 \leq_{\kappa_S}~M$ if $P_+ M P_+$ is a positive operator and $P_- M P_-$ is a negative operator.
\end{defi}

We have the following properties
\begin{enumerate}
\item If $0\leq_{\kappa_S}~M$ then $\tr{M \kappa_S}\geq 0$. 
\item The maps $D_i$ and $F$ are nondecreasing with respect to $\geq_{\kappa_S}$. 
\item If $\beta_1\leq \beta_2$ then $M_{\beta_1}\leq_{\kappa_S}~M_{\beta_2}$. 
\end{enumerate}
The first property is trivial, the second property is a consequence of the fact that $T_S$ and $\Theta_i\Theta_i^*$ commutes with $\kappa_S$ and that $\Theta_i\Theta_i^*$ is a positive operator. The third property is a consequence of the definition of $M_{\beta}$.

Let $\Delta$ be the operator defined by
\[
\Delta=F\left(\sum_{i=k+1}^n D_i( M_{\beta_k}-M_{\beta_{i}}) \right)~.
\]
 Then by the linearity of $F$ we have
\[
M_S+\Delta=F\left(\sum_{i\leq k} D_i(M_{\beta_i})+\sum_{i>k}D_i(M_{\beta_k})\right)~.
\]
 By properties 2 and 3 above we have
\[
\Delta \leq_{\kappa_S} 0~
\]
and by Equation \eqref{eq:zprope_F} we have
\begin{align}\label{eq:zmajoration_M_delta}
M_S+\Delta\leq_{\kappa_S} F\left(\sum_{i=1}^n D_i(M_{\beta_k})\right)=M_{\beta_k}~.
\end{align}
By definition of $F$, we have 
\[
\sum_{i=1}^n D_i(M_S+\Delta)-i[T_S, M_S+\Delta]=\sum_{i\leq k} D_i(M_{\beta_i})+\sum_{i>k} D_i(M_{\beta_k})~
\]
where $-i[T_S, M_S+\Delta]$ satisfies $\tr{-i[T_S, M_S+\Delta]\kappa_S}=0$ since $T_S$ commutes with $\kappa_S$. Thus
\[
\sum_{i=1}^k J_i=\tr{\kappa_S \sum_{i>k} D_i(M_S+\Delta-M_{\beta_k})+\sum_{i\geq k} D_i(\Delta)}~.
\]
In the trace, the first sum is $\leq_{\kappa_S} 0$ because of Equation \eqref{eq:zmajoration_M_delta} and the second sum is $\leq_{\kappa_S} 0$ because $\Delta \leq_{\kappa_S}0$. This concludes the proof in the case where $(\Lambda^t)$ is positivity improving. 
\vspace{0.5cm}

In the general case, we can take the limit for some perturbation of the $\Theta_i$ which makes the map positivity improving. Alternatively, let us decompose $\yy=V_1\oplus V_2$ where $V_1=K(T_S, \Theta)$ and $V_2=V_1^\perp$. Note that $range (\Theta_i)\subset V_1$ for any $i$, and $V_1$ is stable by $T_S$; let us decompose $M_S$ according to the decomposition of $\yy$, of the form
\[
M_S=\begin{pmatrix} M_{11} & M_{12} \\ M_{12}^* & M_{22} \end{pmatrix}
\]
where $M_{11}$ acts only on $V_1$ and so on, and decompose $K_S$ into blocks written $K_{i,j}$ the same way. Then the fact that $K_S$ commutes with $\Theta_i \Theta_i^*$ implies that $\Theta_i\Theta_i^* K_{12}=0$ and so 
\[
J_i=\tr{D_i([M_{\beta_i}]_{11}-M_{11})}.
\]
Thus it is sufficient to consider the restriction on the space of matrices on $V_1$, which is preserved by the map $M\rightarrow GM+MG^*$ and on which this map is invertible; the proof in the positivity improving case applies.
\end{proof}

\subsubsection{Gauge-invariant quasi-free fermionic semigroups}\label{subsub:gauge_invariant}

In the context of fermionic system, Gauge invariance means commuting with the number operator $N=\sum_{i=1}^L c_i^*c_i$. This property really depends on the subspace $\hh_0$ of $\yy$, and not just on the couple $(\yy, \xi)$. The properties of Gauge-invariant operators are best described in the creation/annihilation basis. In what follows we use the \enquote{small} row operator $C^*: \hh_0 \otimes \hh \rightarrow \hh$ defined by 
\[
C^*(x\otimes \ket{u})=c^*_{x}\ket{u}~.
\]
Any quadratic gauge-invariant operator can be written $\lambda \Un+ C^* T^0 C$ for some operator $T: \hh_0\rightarrow \hh_0$ and some constant $\lambda \in \C$. This kind of operator can be interpreted as acting independently on each fermionic particle, with no interactions between them. Similarly, any gauge-invariant state has a covariance matrix which is block-diagonal in the creation/annihilation basis; since $\cov(\rho)+\xi \cov(\rho)\xi=\Un$ the covariance matrix is of the form
\begin{align*}
\cov(\rho)=\begin{pmatrix} \cov_0 (\rho) & 0 \\ 0 & \Un-\cov_0(\rho) \end{pmatrix}
\end{align*}
where $\cov_0(\rho)=\tra_{\hh}\left(\rho C C^*\right)$ will be called the \enquote{small covariance matrix} in this article. Any gauge-invariant quasi-free state is fully described by its small covariance matrix. Gauge-invariant quasi-free fermionic semigroups can be fully described only in terms of operators on $\yy$: 

\begin{prop}
Let $(\Lambda^t)_{t\geq 0}$ be a quasi-free fermionic semigroup as above, and assume that the operators $H_S$ and $H_{SB}$ and the state $\rho_B$ are gauge-invariant. Let $T_S^0, \Theta^0$ be such that
\begin{align*}
H_S&=C_S^* T_S^0 C_S +\lambda_S \Un & H_{SB}&=C_S^* \Theta^0 C_B+C_B^* \left(\Theta^0\right)^* C_S +\lambda_{SB} \Un
\end{align*}
and let $M_B^0=\cov_0(\rho_B)$. Then the small covariance matrix $M_S^0(t)=\cov_0(\Lambda^t(\rho))$ satisfies the equation
\begin{align}
\frac{d}{dt} M_S^0(t)=\left(-iT^0_S-\frac{1}{2} \Theta^0 \left(\Theta^0\right)^*\right) M_S^0(t)+M_S^0(t)\left(iT^0_S-\frac{1}{2}\Theta^0 \left(\Theta^0\right)^*\right)+\Theta^0 M_B^0 \left(\Theta^0\right)^*~.
\end{align}
The map $\Phi$ writes
\[
\Phi(A)= \sum_{i, j}  \left[\Theta^0 M_B^0 \left(\Theta^0\right)^*\right]_{i,j} c_i^* A c_j +\left[\Theta^0 (\Un-M_B^0) \left(\Theta^0\right)^*\right]_{i,j} c_i A c_j^*~.
\]
If the semigroup is thermal with gauge-invariant conserved quantity $K_S=C_S \kappa_S^0 C_S$ then the flux of 
energy entering the i-th baths is 
\[
J_i(\rho)=\tr{\Theta^0 \left(\Theta^0\right)^*\left(M_{\beta_i}^0-M_S^0\right)}~
\]
where $M_{\beta_i}^0=\left(1+e^{-\beta_i \kappa_S^0}\right)^{-1}$.
\end{prop}
Note that the number of fermions $N$ is always a globally conserved quantity in a gauge-invariant quasi-free fermionic system. However, the states $\rho_{Bi}$ of the sub-baths are not always thermal with respect to the number operator $N_{Bi}$, so we cannot always take $K_S=N_S$. 
\vspace{0.5cm}

{\bf The example of the fermionic chain :} Let us treat an example where $K_S=N_S$: the fermionic chain. We take two baths, indexed by $0$ and $L+1$, and put the $L$ sites of the system \enquote{between them}. 
We choose $\hh_{B0}$ and $\hh_{B(L+1)}$  with one site each, with energies $K_{0}=N_{B0}=c_{B0}^* c_{B 0}$ and $K_{L}=N_{BL}=c_{B (L+1)}^* c_{B(L+1)}$   at temperature $\beta_1, \beta_L$, and consider
\begin{align*}
\Theta^0&=
\begin{pmatrix}
\theta_0 & 0 \\
0 & 0 \\
\vdots & \vdots \\
0 & \theta_{L+1}
\end{pmatrix}\\
T_{S}^0&= D+D^T
\end{align*}
where $D$ is the upper-diagonal matrix
\[
D=\begin{pmatrix}
0 & 1 &  & \\
0 & 0 & 1 & & 0\\
& 0 & 0 & 1 & \\
 & &\ddots & \ddots & \ddots 
\end{pmatrix}~.
\]
Thus every site of the bath is in contact only with the nearest sites, with intensity $1$ inside of the system and intensity $\theta_0, \theta_{L+1}$ at the interface between the system and the baths. 

This system is positivity improving; and the stationary state can be described explicitly (see \cite{AndreysFermions1} for a more detailed treatment). Let us write $n_0=(1+e^{-\beta_0})^{-1}$ and $n_{L+1}=(1+e^{-\beta_1})^{-1}$ the unique elements of the (small) covariance matrices of the baths. Then the small covariance matrix $M_\infty^0$ of the stationary state $\rho_\infty$ is of the form
\begin{align*}
M_\infty^0&=\begin{pmatrix}
p_1 & ij & 0 & ... \\
-ij & p_m & ij ... \\
0& -ij & p_m & ... \\
&&& \ddots& \\
0 & ...& &-ij &p_m & ij\\
0 & ... &&& -ij & p_L
\end{pmatrix}
\end{align*}
where $p_1, p_m, p_L $ and $j$ are real numbers that are independent of $L$. They are defined as follows.
Let $s=4(\theta_0^2+\theta_{L+1}^2)+\theta_1^2\theta_{L+1}^2(\theta_1^2+\theta_{L+1}^2)$. Then 
\begin{align*}
p_0&=\frac{1}{s}\Big(\theta_0^2\left(\theta_{L+1}^4+\theta_0^2\theta_{L+1}^2+4\right)n_0 +4\theta_0^2n_{L+1}\Big) \\
p_m&=\frac{1}{s}\Big(\theta_0^2\left( \theta_{L+1}^4+4\right)n_0+\theta_{L+1}^2\left(\theta_0^4 +4\right)n_{L+1}\Big) \\
p_{L+1}&= \frac{1}{s}\Big( 4\theta_{L+1}^2n_0+\theta_0^2\left(\theta_{L+1}^4+\theta_{L+1}^2\theta_0^2+4\right)n_{L+1}\Big)\\
j&= \frac{2}{s} \theta_0^2\theta_{L+1}^2 (n_0-n_{L+1}) \,.
\end{align*}

The energy fluxes are also independent of $L$, and
\begin{align}\label{eq:flux_chain}
J_0=-J_1=\theta_0^2 \left(n_0-p_1\right)=2 j ~.
\end{align}

\section{Large deviations of energy exchanges for fermionic semigroups}

In the repeated interaction model, it is possible to measure the energy of each sub-bath before and after the interaction, thus measuring the energy exchanged during this interaction. In the continuous-time limit, this allows to treat the energy fluxes between the bath and the system as random variables, whose average over time converges to the stationary fluxes $J_i$. The purpose of this section is to study the large deviations of these random variables in the case of quasi-free fermionic systems; we first present the large deviation results in the case of thermal models, based on \cite{pellegrini_jumps_10} and \cite{JPW14}, and then describe more precisely the quasi-free fermionic case.  The explicit computation of the large deviation rate involves the resolution of an algebraic Riccati equation, which is reminiscent of \cite{JPS16}.

\subsection{Repeated measurement process for thermal models}

A Markovian evolution under some indirect continual measurement can be described by the notion of \emph{unraveling}:

\begin{defi}
Let $I=[0, +\infty)$ or $I=\N$ be a set of times, and let $(\Omega_t, \nu_t)_{t\in I}$ be a family of standard measured spaces (where $\nu_t$ are Radon measures). We assume that there is a measure-preserving bijection 
\[
\begin{array}{lll}
(\Omega_t\times \Omega_s,~ \nu_t\otimes \nu_s)&\rightarrow &(\Omega_{t+s},~ \nu_{t+s}) \\
(\omega_t, ~ \omega_s) & \mapsto &\omega_t.\omega_s
\end{array}
\]
 and for $\omega\in \Omega_{t+s}$ we write $\omega_{t]}$ the projection of $\omega$ on $\Omega_t$, i.e; the event such that $\omega=\omega_{t]}\cdot\omega_{(t,t+s]}$ for some event $\omega_{(t,s]}\in \Omega_s$. Let $\hh_S$ be a separable Hilbert space, and $\call$ the generator of a QMS $(\Lambda^t)_{t\in I}$. A \emph{Markovian unraveling} on $\hh_S, \Omega_t$ is a measurable map 
\[
(t\in I, \omega\in \Omega_t) \mapsto \Psi_t[\omega]
\]
where $\Psi_t[\omega]$ is a completely positive map for a.e. $\omega\in \Omega_t$, with the following properties: 
\begin{enumerate}
\item For all $t\in I$ we have
\[
\int_{\Omega_t} d\nu_t(\omega) \Psi_t[\omega]=\Lambda^t~.
\]
\item For all $s, t\in I$ and $\omega_t\in \Omega_t, \omega_s\in \Omega_s$, we have
\[
\Psi_{t+s}[\omega_t.\omega_s]=\Psi_t[\omega_t]\circ\Psi_s[\omega_s]~.
\]
\end{enumerate}

To any Markovian unraveling, any initial state $\rho$ and any $T$ corresponds a probability measure $\P_T$ on $\Omega_T$ and a process $(\tilde{\rho}_t)_{0\leq t\leq T}$ defined by
\begin{align}
d\P_T(\omega)&=\tr{\rho\Psi_T[\omega](\Un)}d\nu_T(\omega) \\
\tilde{\rho}_t&=\frac{\Psi_t[\omega_{t]}]^*(\rho)}{\tr{\Psi_t[\omega_{t]}]^*(\rho)}}~.
\end{align}

\end{defi}

The first property ensures that $\E_{\P_T}(\tilde{\rho}_t)=\left(\Lambda^t\right)^*(\rho)$, while the second property ensures that $\P_t$ is the pushforward measure of $\P_T$ under the projection $\omega\mapsto \omega_{t]}$. In discrete time, such an unraveling can be obtained by performing measures on the bath, as follows.
\vspace{0.5cm}

Let us consider a thermal model as above (subsection \ref{subseq:energy_conservation}). Let us write $P_{i, E_i}$ the projector on the eigenspace of $K_{B_i}$ for eigenvalue $E_i$. For any list of eigenvalues $E=(E_1, \cdots, E_n)$ consider the projector
\[
P_E=\bigotimes_{i=1}^n P_{i, E_i}~.
\]

If we perform one interaction and we measure $K_{B_i}$ before and after the interaction, with initial state $\rho$, the state after the interaction is 
\[
\tilde{\rho}_S(\tau)=\frac{{\Psi}_{\tau}^*[E, F](\rho)}{\tr{{\Psi}_{\tau}^*[E, F](\rho)}}
\]
where $E=(E_1, \cdots, E_n)$ is the result of the first measurement of $(K_{B_1}, \cdots, K_{B_n})$, where $F$ is the result of the second measurement of the $K_{B_i}$, and
\[
{\Psi}_{\tau}^*[E, F](\rho)=\tra_{B}\left(P_F U_\tau ~\rho\otimes \left(P_E\rho_B P_E \right)~ U_\tau^* P_F \right)~.
\]
The outcome $(E, F)$ of the measurement appears with probability $\tr{{\Psi}_{\tau}^*[E, F](\rho)}$. Since we are interested only in the fluxes, we can forget about the precise outcome $(E, F)$ and just retain the difference $\delta^i=F_i-E_i$. The state is then 
\[
\frac{\Psi_\tau^*[\delta](\rho)}{\tr{\Psi_\tau^*[\delta](\rho)}}
\]
where 
\[
\Psi_\tau^*[\delta](\rho)=\sum_{E, F \text{ with } F-E=\delta} \Psi_\tau^*[E, F](\rho)~.
\]

Applying this measurement and interaction repeatedly, we obtain a random process $(\delta_\tau(k))_{k\in \N}$ coupled with a random process of states $(\tilde{\rho}_{\tau}(k))_{k\in \N}$. It is a Markovian unraveling of the QMS $(\Lambda_\tau^k)_{k\in \N}$ of repeated interaction. Indeed, write 
\[
N_\tau^i(k)=\sum_{i=1}^k \delta^i_\tau(i)
\]
the total energy exchange up to time $k\tau$, and $N_\tau(k)=(N_\tau^1(k), \cdots, N_\tau^n(k))$. For $t\in \N$ the trajectory $(N_\tau(k))_{k\leq t}$ is in the space $\Omega_{\tau, t}=\left(\R^n\right)^t$ which we endow with the counting measure $\nu_t$~. Then the process $\tilde{\rho}_{\tau}(k)$ is a Markovian unraveling of $\Lambda_\tau^k$ with maps
\[
 \Psi_{\tau,t}[(N_\tau(k))_{k\leq t}]=\Psi_\tau[N_\tau(1)-N_\tau(0)]\circ \cdots \circ \Psi_\tau[N_\tau(t)-N_\tau(t-1)]~.
\]

The following theorem describes the limit in distribution of this process as $\tau\rightarrow 0$. It is a generalization of a theorem of Nechita and Pellegrini \cite{Pellegrini_Nechita09}.

\begin{theo}[Nechita and Pellegrini]
Suppose that Assumption \ref{ass:strongConditionV} is satisfied. Then for any $T>0$ the process 
\[
(\tilde{\rho}_\tau(\ent{t/\tau}), N_{\tau}(\ent{t/\tau}))_{t\in [0, T]}
\]
converges in distribution in the space of c\`adl\`ag functions to a process $(\tilde{\rho}_t, N_t)_{t\in [0, T]}$ where $t\mapsto N_t^i$ is piecewise constant, with a finite number of jumps, all in the set $ sp(N_{B_i})-sp(N_{B_i})$. The distribution of this process is described as a Markovian unraveling, the following way: we describe any trajectory by the list of jumps, so our universe $\Omega_T$ is

\[
\Omega_T=\left\{ \Big((t_1, i_1, \delta_1), \cdots, (t_k, i_k, \delta_k)\Big)~|~ 0<t_1<\cdots <t_k<T ~,~\delta_l \in sp(N_{B_{i_l}})-sp(N_{B_{i_l}})~,~ k\in \N\right\}
\]
where $t_l$ is the time of the $l$-th jump, $i_l$ is the number of the reservoirs on which the jump appears, and $\delta_l$ is the energy exchanged during the jump. All the parameters except the $t_l$ are on a discrete set; we endow $\Omega_T$ of the measure
\[
d\nu_T\Big((t_1, i_1, \delta_1), \cdots, (t_k, i_k, \delta_k)\Big)=\Un_{0\leq t_1\leq \cdots \leq t_k}dt_1\cdots dt_k
\]
(where $d t_l$ is the Lebesgue measure).

The map $\Psi_t[\omega]$ is defined the following way: 

 For any $\delta^i\in sp(N_{B_i})-sp(N_{B_i})$ consider the completely positive map $\Phi_{i, \delta_i}$ defined by
\[
\Phi_{i,\delta^i}^*(\rho)=\sum_{E_i, F_i \text{ with } F_i-E_i=\delta^i} \tra_{B_i} \left(P_{i, F_i} H_{SB_i} P_{i, E_i} \left(\rho\otimes \rho_{B_i}\right) P_{i, E_i} H_{SB_i} P_{i, F_i} \right)
\]
(thus $\sum_{i}\sum_{\delta_i} \Phi_{i, \delta^i}=\Phi$). Write $G=-iH_S-\frac{1}{2} \Phi(\Un)$ and define the map
\[
\call_G^*(A)=GA+AG^*~.
\]
For any trajectory $\omega=((t_1, i_1, \delta_1), \cdots, (t_k, i_k, \delta_k))$, we define
\begin{align}\label{eq:defi_unraveling}
\Psi_T[\omega]^*=e^{(T-t_k) \call_G^*} \Phi_{i_k, \delta_k}^* e^{(t_k-t_{k-1})\call_G^*} \Phi_{i_{k-1}, \delta_{k-1}}^* \cdots \Phi_{i_1, \delta_1}^* e^{t_1 \call_G^*}~.
\end{align}
Then $(N^i_t)_{t\in [0, T]}$ is the random variable on $\Omega_T, d\P_T(\omega)=\tr{\Psi_t[\omega]^*(\rho)}d\nu_t(\omega)$ defined by the map
\[
N^i_t[\big(t_1, i_1, \delta_1), \cdots, (t_k, i_k, \delta_k)\Big)]=  \sum_{l=1}^k \Un_{i_l=i}\Un_{t_l\leq t} \delta_l~
\]
and $\tilde{\rho}_t$ is the random variable 
\[
\Psi_t[\omega_{t]}]^*(\rho)~.
\]
\end{theo}

The Dyson expansion formula applied to $\call^*=\call_G^*+\sum_{i, \delta^i} \Phi_{i, \delta_i}^*$ gives
 
\begin{align*}
\int_{\Omega_T}\Psi_T[\omega]^* d\nu_t(\omega)
&=\sum_{n=0}^\infty \sum_{i_1, \delta_1, \cdots, i_k, \delta_k} \int_{0\leq t_1\leq \cdots \leq t_k\leq T} e^{(T-t_k) \call_G^*} \Phi_{i_k, \delta_k}^* e^{(t_k-t_{k-1})\call_G^*} \Phi_{i_{k-1}, \delta_{k-1}}^* \cdots \Phi_{i_1, \delta_1}^* e^{t_1 \call_G^*} dt_1\cdots dt_n \\
&=e^{T\call^*}~
\end{align*}
so $\omega\mapsto \Psi_T[\omega]$ fits the definition of a Markovian unraveling.
\newline

 Another way to describe $\Psi_T[\omega]$ is to say that 
$h_t=\Psi_t^*[\omega](\rho)$ satisfies the following stochastic differential equation with jumps
\[
dh_t=\call_G^*(h_t) dt +\sum_{i=1}^n \Phi^*_{i, dN^i_t}(h_{t_-})~.
\]

\begin{proof}[Elements of proof:]
The proof of this theorem is the same as the one of Theorem 5 of \cite{Pellegrini_Nechita09}, which was restricted to bath of dimension 2. We will not reproduce it here, let us just present the heuristics. We obtained the following estimate using Assumption \ref{ass:strongConditionV}:
\begin{align}\label{eq:estimate_random}
\Psi_\tau[\delta]^*(\rho)=\left\{\begin{array}{ll}
\rho+\tau\call_G^*(\rho)+o(\tau) & \text{if $\delta^i=0$ for all $i$} \\
\tau  \Phi_{k, \delta^k}^*(\rho)+o(\tau) & \text{if $\delta^i= 0$ for all $i\neq k$ and $\delta^k\neq 0$}\\
o(\tau) & \text{if two of the $\delta^i$'s are nonzero} 
\end{array}\right.
\end{align}
where the $o(\tau)$ is uniformly small on $[0, T]$. Thus, the probability to have a jump during a given interaction is of order $\tau$. This implies that most of the time, $\delta=0$ and $\Psi_\tau[\delta]^*(\rho)=\call_G^*(\rho)$. Exactly one of the bath exchange energy with a probability of order $\tau$, and two of the baths exchange energy only with a probability of order $o(\tau)$, so it becomes negligible in the limit $\tau\rightarrow 0$. The complete proof involves showing the tightness of the process as $\tau\rightarrow 0$ and proving that the limit process is the only solution of a martingale problem, see \cite{pellegrini_jumps_10} or \cite{Pellegrini_Nechita09}.
\end{proof}

The random process is related to the deterministic quantum dynamic as follows: 

\begin{theo}
Let us write $\rho(t)=(\Lambda^t)^*(\rho_0)$. Then we have
\begin{align*}
\E_{\P_T}(\tilde{\rho}_T)&=\rho(T)\\
\E_{\P_T}(N_T^i)&=\int_0^T J_i(\rho(t))dt\\
\end{align*}
where $J_i(\rho(t))$ is the flux of energy entering the $i$-th bath when no measurement is performed, defined in Section \ref{subseq:energy_fluxes}.
\end{theo}
\begin{proof}
The first equation is a consequence of the property $\int_{\omega_T} d\nu_T(\omega)\Phi_T[\omega]=e^{t\call}$. 

For the second fact, we use the discrete-time version: by definition of the quantum measurement, we have
\begin{align*}
\E(N_{\tau}^i(\ent{T/\tau})=D_{\tau, i, 0\rightarrow \ent{T/\tau}}(\rho)
\end{align*}
and as in subsection \ref{subseq:energy_fluxes} the right-hand side converges to $\int_0^T J_i(\rho(t)dt$. To prove that the left-hand side converges to $\E_{\P_T}(N_T^i)$ we use the convergence in distribution; the random variable $N_{\tau}^i(\ent{T/\tau})$ is not bounded, so we need to bound the probability that it is large. For this, consider the random variable $Z_{\tau,T}$ equal to the number of jumps before time $T$: 
\begin{align}
Z_{\tau,T}=\#\set{k \leq \ent{T/\tau} | ~N_\tau((k+1))-N_\tau(k)\neq 0}~.
\end{align}
and its continuous-time version $Z_T(((t_1, i_1, \delta_1), \cdots, (t_k, i_k, \delta_k)))=k$. Then, for any $K\in \N$ the random variable $N_\tau^i(\ent{T/\tau})\Un_{Z_{\tau, T}\leq K}$ is bounded, hence its expectancy converges to $\E(N_T^i \Un_{Z_T\leq K})$. To show that we can neglect the rest, we use the following: 
\begin{lem}
There exists a constant $C$ such that for any $T$, and for any $k\in \N$, 
\[
\P(Z_T= k)\leq \frac{T^k C^k}{k!}
\]
\end{lem}

\begin{proof}
Note that $\tr{e^{t\call_G^*}(\rho)}\leq \tr{\rho}$ for any $t$; let $C_1$ be greater than any of the operator norms of the $\Phi_{i, \delta^i}^*$ and let $C_2$ be the number of couples $(i, \delta^i)$ possible. Then 
\begin{align*}
\P_T(Z_T=k)&=\sum_{(i_1, \delta_1), \cdots, (i_k, \delta_k)} \int_{0\leq t_1\leq \cdots \leq t_k}  \tr{ e^{(T-t_n) \call_G^*} \Phi_{i_n, \delta_n}^* \cdots \Phi_{i_1, \delta_1}^* e^{t_1 \call_G^*}(\rho_0)} dt_1\cdots dt_n \\
&\leq  \frac{C_2^k C_1^k T^k}{k!}~.
\end{align*}
\end{proof}

By the convergence in distribution on $[0, T]$, we have 
\[
\P(Z_{\tau, T}=k)\leq\frac{ T^{k} C_3^k}{k!}
\]
 for all $\tau$ small enough, for some constant $C_3$ slightly larger than $C$. Note that
\[
\abs{N^i_T} \leq Z_{\tau, T} \max \delta^i~.
\]
Thus the rest 
\[
\abs{\E\Big(N_\tau^i(\ent{T/\tau})\Big)-\E\Big(N_\tau^i(\ent{T/\tau})\Un_{Z_{\tau, T}\leq K}\Big)}\leq \max(\delta^i) \sum_{k>K} \frac{C_3^k T^k}{k!}
\]
converges to zero as $K\rightarrow +\infty$, uniformly in $\tau$, and the same for the continuous-time version.
\end{proof}

\subsection{Large deviation of the energy fluxes}

In this section, we consider the large deviations of the random variables $\frac{N_T^i}{T}$ as $T\rightarrow \infty$. We assume that $\call$ is positivity improving, thus it has a unique stationary state $\rho_\infty$, and
\[
\lim \E_{\P_T} \left(\frac{N_T^i}{T}\right)=J_i(\rho_\infty)~.
\]
This section is entirely based on Jak\v{s}i\'c, Pillet and Westrich \cite{JPW14}. In this article, the authors construct the random variables $N_T^i$ directly from the semigroup, and study its large deviations. Our notations differ from their by one notable point: the authors express the large deviation of the entropy exchange $\beta_i N_T^i$, while we consider the large deviations of the energy exchange $N_T^i$. We just recall their main results, giving only elements of proofs.
\vspace{0.5cm}

In the case of finite-dimensional systems, the large deviations of $N_T/T$ fall into the simplest case of the G\"artner-Ellis theorem: it is sufficient to study the limit 
\begin{align}\label{eq:def_ealpha}
e(\alpha)=\lim_{T\rightarrow \infty} \frac{1}{T} \log\E\left(e^{\scal{ \alpha, N_T}}\right)~
\end{align}
for $\alpha=(\alpha_1, \cdots, \alpha_n)\in \R^n$ (and where $\scal{ \alpha, N_T}=\sum_{i=1}^n \alpha_i N_T^i$.  The authors of \cite{JPW14} prove the following large deviation principle: 

\begin{theo}[Jak\v{s}i\'c, Pillet, Westrich] \label{theo:large_deviations}
There exists a convex and continuous rate function $I$ on $\R^n$ such that for any open set $E\subset \R^n$ we have
\[
\lim_{T\rightarrow \infty} \frac{1}{T}\log\P\left(\frac{1}{T} N_T \in E\right)=-\inf_{\zeta \in \E} I(\zeta)
\]
The function $I$ vanishes only at $\zeta=(J_1, \cdots, J_n)$, it satisfies
\begin{align}
&I(\zeta)=+\infty~~ \text{   if $\sum_{i=1}^n \zeta_i \neq 0$} \label{eq:I_sum}\\
&I(\zeta)-I(-\zeta)=\sum_{i=1}^n \beta_i \zeta_i  \label{eq:I_sym}
\end{align}

 and it is the Legendre transform of $\alpha\mapsto e(\alpha)$: 
\begin{align}\label{eq:legendre}
I(\zeta)=\sup_{\alpha\in \R^n} \big(\scal{\alpha, \zeta}-e(\alpha)\big)~.
\end{align}
\end{theo} 

Relation \ref{eq:I_sum} is a manifestation of the first principle (the sum of the fluxes must be $0$), and has been remarked the first time in \cite{AndrieuGMT_09_fluctuations} (Proposition 1) while relation \ref{eq:I_sym} is linked to the second principle, and is called the Gallavotti-Cohen symmetry. The quantity $\sum_{i=1}^n \beta_i \zeta_i$ is interpreted as the mean entropy production, and heuristically the large deviations principle says that
\[
\frac{d\P(N_T=\zeta)}{d\P(N_T=-\zeta)}\simeq e^{-T\sum_{i=1}^n \beta_i \zeta_i}
\]
as $T\rightarrow \infty$.

This theorem is the consequence of the G\"artner-Ellis theorem and the following properties of $e(\alpha)$: 

\begin{prop}[Jak\v{s}i\'c, Pillet, Westrich]
For any $\alpha\in  \R^n$ the limit \ref{eq:def_ealpha} exists, and it satisfies the following properties: 
\begin{enumerate}
\item The function $\alpha \mapsto e(\alpha)$ is convex and real analytic on $\R^n$. 
\item For any $\alpha \in \R^n$ and $\lambda\in \R$, writing $1_n=(1, \cdots, 1)$ we have
\[
e(\alpha+\lambda 1_n)=e(\alpha)~.
\]
\item \label{prope:Gallavoti_cohen} For any $\alpha\in \R^n$, writing $\beta=(\beta_1, \cdots, \beta_n)$ we have
\[
e(\alpha-\beta)=e(-\alpha)~.
\]
\item For all $i$ we have
\[
\left.\frac{\partial e(\alpha)}{\partial \alpha_i}\right|_{\alpha=0}=-J_i~.
\]
\end{enumerate}
\end{prop}

The proof of these properties goes through the study of the \enquote{deformed semigroup}, of generator $\call_\alpha^*$, as follows: 

\begin{prop}[Jak\v{s}i\'c, Pillet, Westrich]
We define the super-operator $\call_\alpha$ by
\begin{align}\label{eq:def_deformed_semi}
\call_\alpha(A)&=\call_G(A)+\sum_{i=1}^n \Phi_i\left(A e^{-\alpha_i K_S}\right)e^{\alpha_i N_S}
\end{align}
Then
\begin{enumerate}
\item For any $T\in \R$ and $\alpha\in \R^n$ we have
\[
\E\left(e^{\scal{\alpha,  N_T}}\right)=\tr{e^{T\call_\alpha^*}(\rho_0)}~.
\]
\item For all $\alpha\in \R^n$ the super-operator $e^{t\call_\alpha}$ is positivity improving. In particular the dominant eigenvalue of $\call_\alpha^*$ is real and of multiplicity $1$, and the corresponding eigenvector is a positive operator.
\item  $e(\alpha)$ is the dominant eigenvalue of $\call_\alpha$. 
\item The super-operator $\call_\alpha$ is equal to
\begin{align}
\call_\alpha(A)&=\call_G(A)+\sum_{i=1}^n \sum_{\delta^i} e^{-\alpha\delta^i} \Phi_{i, \delta^i}(A)  \\
&=\call_G(A)+\sum_{i=1}^n e^{\frac{\alpha_i}{2} N_S} \Phi_i\left( e^{-\frac{\alpha_i}{2} K_S}A e^{-\frac{\alpha_i}{2} K_S}\right)e^{\frac{\alpha_i}{2} N_S}
\end{align}
\end{enumerate} 
\end{prop}

The third assertion implies the analyticity of $\alpha\rightarrow e(\alpha)$, since $\alpha\rightarrow \call_\alpha^*$ is analytic and the dominant eigenvalue of $\call_\alpha^*$ is simple for all $\alpha$.

\begin{proof}[Elements of proof:]
For the first part, the definition of $N_T$ and of $\P_T$ as unraveling of $\Lambda^t$ and Formula \ref{eq:defi_unraveling} allows to express $\E\left(e^{\scal{\alpha,  N_T}}\right)$ as a sum of multiple integrals, which happens to be the Dyson expansion of $\call_\alpha$. Another derivation of this formula goes through the discrete-time limit: in the repeated interaction procedure with interaction time $\tau$, we have
\begin{align*}
\E\left(e^{\scal{\alpha, N_\tau(k)}}\right)&=\sum_{(N_\tau(l))_{l\leq k}\in \Omega_{\tau, k}} e^{\scal{\alpha, N_\tau(k)}} \tr{\rho\Psi_{\tau, k}\left[(N_\tau(l))_{l\leq k}\right](\Un)} \\
&= \tr{\rho\left(\prod_{l=1}^k \sum_{\delta_\tau} e^{\scal{\alpha, \delta_\tau}} \Psi_\tau[\delta_\tau]\right)(\Un)  }~.\\
\end{align*}
  Write $\alpha\cdot K_B=\sum_{i=1}^n \alpha_i K_{B_i}$ and define $\Lambda_{\tau, \alpha}$ by
\begin{align}\label{eq:discrete_l_alpha}
\Lambda_{\tau, \alpha}^*(\rho)=\sum_{\delta_\tau} e^{\scal{\alpha, \delta_\tau}} \Psi_\tau^*[\delta_\tau](\rho)=\tra_{\hh_B} \left(  e^{-\frac{\alpha}{2}\cdot K_B}U_\tau \left(\rho\otimes\rho_B  e^{\alpha \cdot K_B} \right)U_\tau^* e^{-\frac{\alpha}{2}\cdot K_B} \right)~.
\end{align}
Then the following holds
\[
\Lambda_{\tau, \alpha}^*(\rho)=\rho+\tau \call_\alpha^*(\rho)+o(\tau)
\]
so we have
\[
\lim_{\tau\rightarrow 0}\left(\Lambda_\tau\right)^{\ent{t/\tau}}=e^{t\call_\alpha}
\]
and so 
\[
\lim_{\tau\rightarrow 0}\E\left(e^{\scal{\alpha, N_\tau\left(\ent{t/\tau}\right)}}\right)=\tr{\rho e^{t\call_\alpha(\Un)}}~.
\]

For the second part, since $e^{t\call^*}$ is positivity improving, for any nonzero positive operators $A, B$ on $\hh_S$ we have
\begin{align*}
\tr{A e^{t\call}(B)}>0~.
\end{align*}
Applying the Dyson formula we get
\begin{align*}
\tr{A e^{t\call}(B)}=\sum_{k=0}^{+\infty} \sum_{i_1, \delta_1, \cdots, i_n, \delta_n} \int_{0\leq t_1<\cdots < t_n < t} \tr{A e^{(t-t_n)\call_G}\Phi_{i_n, \delta_n}\cdots e^{t_1\call_G}(B)}~.
\end{align*}
This implies that one of the terms in this sum is strictly positive. Moreover, 
\begin{align*}
\tr{A e^{t\call_\alpha}(B)}=\sum_{k=0}^{+\infty} \sum_{i_1, \delta_1, \cdots, i_n, \delta_n} \exp\left(-{\sum_{i=1}^n \alpha_i \delta^i} \right)\int_{0\leq t_1<\cdots < t_n < t} \tr{A e^{(t-t_n)\call_G}\Phi_{i_n, \delta_n}\cdots e^{t_1\call_G}(B)}~.
\end{align*}
All the terms in this sum are nonnegative, and one of them is strictly positive, so it is strictly positive. This implies the complete positivity of $e^{t\call_\alpha}$. 

The third assertion is a consequence of two first assertions: if $\lambda_\alpha$ is the dominant eigenvalue of $e^{t\call_\alpha^*}$, then the corresponding eigenvector $\rho_\alpha$ is positive definite (since $e^{t\call_\alpha^*}$ is positivity improving), and we may assume that it is of trace $1$. There there exists $\eps>0$ such that $e^{\call_\alpha^*}(\rho_0)\geq  \eps \rho_\alpha$ hence as $t\rightarrow \infty$,
\[
\tr{e^{t\call_\alpha^*}(\rho_0)}\geq \eps e^{t\lambda_\alpha}+o(e^{t\lambda_\alpha})
\]
which implies the third assertion. 

For the last assertion, we just remark that
\begin{align*}
\Phi^*_{i, \delta^i}(A e^{-\alpha K_S}))&=\sum_{E_i, F_i,~F_i-E_i=\delta^i} \tra_{B_i}\left( H_{SB_i} A e^{-\alpha K_S} P_{F_i} H_{SB_i} P_{E_i} \rho_{B_i}\right)\\
&=\sum_{E_i, F_i,~F_i-E_i=\delta^i}\tra_{B_i}\left( H_{SB_i} Ae^{-\alpha (K_S+K_{B_i})}e^{\alpha K_{B_i}} P_{F_i} H_{SB_i} P_{E_i} \rho_{B_i}\right) \\
&=\sum_{E_i, F_i,~F_i-E_i=\delta^i}\tra_{B_i}\left( H_{SB_i} A e^{\alpha K_{B_i}} P_{F_i} H_{SB_i} P_{E_i} e^{-\alpha (K_S+K_{B_i})}\rho_{B_i}\right) \\
&=\sum_{E_i, F_i,~F_i-E_i=\delta^i}\tra_{B_i}\left( H_{SB_i} A e^{\alpha F_i } P_{F_i} H_{SB_i} P_{E_i} e^{\alpha E_i}\rho_{B_i}\right)e^{\alpha K_S}
\end{align*} 
the last line being obtained by $P_{E_i} e^{-\alpha K_{B_i}}=e^{-\alpha E_i}P_{E_i}$. This gives the first reformulation, the second comes from the fact that the $\Phi_i$ satisfy the detailed balance, so it commutes with the modular operator $\Delta_{e^{-\alpha K_S/2}}$. 
\end{proof}

\subsection{The quasi-free fermionic case}

In this subsection we apply the formalism described above to the case of a quasi-free fermionic system. Hence, we consider a thermal quasi-free fermionic semigroup which is positivity improving; the idea is then to study the maximal eigenvalue of the deformed generator $\call_\alpha$. As shown above, the study of $\call$ is greatly simplified by the existence of a closed equation for its covariance matrix, and the fact that it preserves the set of quasi-free states. In the case of $\call_\alpha$, the covariance matrix does not satisfies a closed equation in general, but the set of multiples of quasi-free states is still preserved, and restricted on this set the covariance matrix evolves according to a closed equation, admittedly more complex than the affine equation of the non-deformed semigroup. This allows to reduce the computation of $e(\alpha)$ to the resolution of an algebraic Riccati equation; the outcome of this study is the following theorem:

\begin{theo}\label{theo:e_alpha_fermionic}
Let us consider a thermal quasi-free fermionic semigroup $\call$ defined as in Paragraph \ref{subsub:qf_semi}, and assume that it is positivity improving. For any $\alpha, \beta\in \R^n$ consider the operators on $\yy$
\begin{align*}
A_\beta&=-iT_S+\sum_{i=1}^n \left(M_{\beta_i}-\frac{1}{2}\Un\right)\Theta_i \Theta_i^*\\
B_{\alpha, \beta}&=\sum_{i=1}^n e^{\alpha_i \kappa_S}M_{\beta_i} \Theta_i \Theta_i^*
\end{align*}
where we wrote $M_{\beta_i}=\left(\Un+e^{-\beta_i \kappa_S}\right)^{-1}$ the covariance matrix of the Gibbs state at temperature $\beta_i$.

Define the operator $Z_\alpha$ on $\yy\oplus \yy$ by
\[
Z_\alpha=\begin{pmatrix} A_\beta & B_{\alpha, \beta} \\ B_{-\alpha, -\beta} & -A_\beta^* \end{pmatrix}~.
\]
The set of eigenvalues of $Z_\alpha$ is symmetric with respect to the imaginary axis, and the pure imaginary eigenvalues are of even multiplicity; let $\lambda_1(\alpha), \cdots, \lambda_{k}(\alpha)$ be its eigenvalues of positive real part. Then 
\[
e(\alpha)=\frac{1}{2} \sum_{i=1}^{k} \lambda_i(\alpha)-\frac{1}{4}\sum_{i=1}^n \tr{\Theta_i \Theta_i^*}~.
\]
\end{theo}

We first prove the following proposition: 

\begin{prop}\label{prop:riccati_e_alpha}
 For any $\alpha \in \R^n$, the deformed semigroup 
\[
t\mapsto e^{t\call_\alpha^*}
\]
preserves the vector space generated by the quasi-free state, and the eigenvector for the  maximal eigenvalue $e(\alpha)$ is proportional to a quasi-free state.

Moreover, any quasi-free state $\rho$ of covariance matrix $M$ is an eigenvector of $\call_\alpha^*$ if and only if $M$ is a solution to the Riccati equation
\begin{align}\label{eq:riccati}
G_{\alpha, \beta} M + M G_{\alpha, \beta}^* + M C_{\alpha, \beta} M +B_{\alpha, \beta}=0
\end{align}

where $G_{\alpha, \beta}$ and $C_{\alpha, \beta}$ are defined the following way: let
\begin{align}
Q_{\alpha, \beta} &=\sum_{i=1}^n \Theta_i \Theta_i^* M_{\beta_i}\left(e^{\alpha_i \kappa_S}-1\right)
\end{align}
and 
\begin{align}
G_{\alpha, \beta} &=G-Q_{\alpha, \beta}=-iT_S-\sum_{i=1}^n  \left(\frac{1}{2}+\left(e^{\alpha_i \kappa_S}-1\right) M_{\beta_i} \right) \Theta_i \Theta_i^*\\
C_{\alpha, \beta}&=Q_{\alpha, \beta}-Q_{\alpha, \beta}^T=\sum_{i=1}^n \Big(\left(e^{\alpha_i \kappa_S}+e^{-\alpha_i \kappa_S}-2\right)M_{\beta_i}+\Un-e^{\alpha_i \kappa_S}\Big)\Theta_i\Theta_i^*~.
\end{align}

The corresponding eigenvalue is 
\begin{align}
\lambda=\frac{1}{2}\tr{Q_{\alpha, \beta}^T M}=-\frac{1}{2} \tr{C_{\alpha, \beta} M}+\frac{1}{2}\tr{Q_{\alpha, \beta}}~.
\end{align}
\end{prop}

\begin{proof}

{\bf The vector space generated by quasi-free states is preserved: }

We use the discrete approximation: let us show that for any quasi-free state $\rho$ the operator $\Lambda_{\tau, \alpha}^*(\rho)$ defined at Equation \eqref{eq:discrete_l_alpha} is proportional to a quasi-free state. 

 First, $\rho_B e^{\alpha.K_B}=\frac{1}{Z}\exp\left(\sum_{i=1}^n (-\beta_i+\alpha_i)K_{B_i}\right)$ is proportional to a quasi-free state; thus its tensor product with $\rho$ is also proportional to a quasi-free state. The unitary $U_\tau$ is a Bogoliubov transform, so the following operator is proportional to a quasi-free state.
\[
U_\tau \left(\rho\otimes\rho_B  e^{-\alpha.K_B} \right)U_\tau^*~.
\]
 Moreover, for any quasi-free states $\sigma, \nu$ the operator $\sigma^{\alpha} \nu \sigma^\alpha$ is proportional to a quasi-free state, and the partial trace of a quasi-free state is quasi-free, so
\[
\Lambda_{\tau, \alpha}^*(\rho)
\] 
is proportional to a quasi-free state. Thus
\[
\left(\Lambda_{\tau,\alpha}^*\right)^{\ent{t/\tau}}(\rho)
\]
is quasi-free for any $t$ and $\tau$, and the set of quasi-free states is closed, hence we can pass to the limit as $\tau\rightarrow 0$ so $^{t\call_\alpha^*}(\rho)$ is proportional to a quasi-free state.
\vspace{0.5cm}

{\bf The eigenvector for the maximal eigenvalue is quasi-free: }

This derives from the following lemma: 

\begin{lem}
Let $\alpha\in \R^n \mapsto L_\alpha\in M_{k,k}(\C)$ be a continuous map, write $\lambda_\alpha$ the dominant eigenvalue of $L_\alpha$, and assume that it is simple for all $\alpha$, of eigenvector $x_\alpha$. Assume that there is a closed cone $\qq$ which is stable by $L_\alpha$ for any $\alpha$, and that $x_0\in \qq$. Then for any $\alpha \in \R^n$ we have $x_\alpha\in \qq$.
\end{lem}
Applying this lemma to $\alpha\mapsto \call_\alpha^*$ and $\qq$ the set of operators proportional to quasi-free states gives that $\rho_\alpha$ is quasi-free for all $\alpha\in \R^n$.
\begin{proof}
We assume $\norm{x_\alpha}=1$. 
Up to choosing the right phase for $x_\alpha$ we can also assume that $\alpha\mapsto (\lambda_\alpha, x_\alpha)$ is continuous since $\alpha\mapsto L_\alpha$ is continuous. 
Thus, the set $E=\{\alpha~|~x_\alpha\in \qq\}$ is closed. Let us show that it is open. 
Consider some $\alpha_0\in E$. 
Write $V_\alpha$ the vector space which is stable by $L_\alpha$ and such that $V_\alpha\oplus (\C x_\alpha)=\C^n$. 
Then  $x_{\alpha_0}=\mu_\alpha x_{\alpha}+z_\alpha$ where $z_\alpha\in V_{\alpha}$ and $\alpha\mapsto \mu_\alpha$ is continuous, and nonzero for any $\alpha$ close enough to $\alpha_0$. Since $\lambda_\alpha$ is the maximal eigenvalue, we have
\[
\lim_{n\rightarrow +\infty} \frac{L_\alpha^n x_{\alpha_0}}{\lambda_\alpha^n}=\mu_\alpha x_{\alpha}
\]
and $L_\alpha^n x_{\alpha_0}\in \qq$ for all $n$ so $x_\alpha\in \qq$ when $\mu_\alpha \neq 0$. 
\end{proof}

{\bf Derivation of the equation for $M$ and $\lambda$:}

Let us first describe the action of $\call_\alpha$ more precisely: for all $i\in \{1, \cdots, n\}$, for any observable $A$, we have 
\begin{align*}
 \Phi_i\left(A ~e^{-\alpha_i K_S}\right)e^{\alpha_i K_S}
&=  F^* \left(A \otimes  \Theta_i M_{B_i} \Theta_i^* \right)\left(e^{-\alpha_i K_S}\otimes \Un_{\yy}\right)F e^{\alpha_i K_S}~~\text{by formula \ref{eq:phi}} \\
&= F^* \left(A\otimes \Theta_i M_{B_i} \Theta_i^*\right) \left(\Un_{\hh_S}\otimes e^{\alpha_i \kappa_S}\right) F~~~\text{by Proposition \ref{prop:commutation_formula_exp_F}} \\
&=F^*\left(A\otimes \Theta_i M_{B_i}\Theta_i e^{\alpha_i\kappa_S}\right) F~.
\end{align*}
By the preservation of energy \ref{eq:db_quadratic_operators} we have 
\[
\Theta_i M_{B_i}\Theta_i e^{\alpha_i\kappa_S}=\Theta_i \Theta_i^* M_{\beta_i} e^{\alpha_i \kappa_S}~.
\]

so
\[
(\call_\alpha-\call)(A)=F^*(A\otimes G) F=\sum_{1\leq k,l\leq 2L_S} [Q_{\alpha, \beta}]^f_{k,l} \gamma_k A \gamma_l~.
\]

Let us now consider a quasi-free state $\rho$ of density matrix $M$, and let us assume that it is an eigenvector of $\call_\alpha^*$
\[
\call_\alpha^*(\rho)=\lambda \rho~.
\]
Then we can express $\lambda$ in terms of $M$, indeed
\begin{align*}
\lambda&=\tr{\call_\alpha^*(\rho)}\\
&=\tr{(\call_\alpha^*-\call^*)(\rho)}\\
&=\tr{\rho F^* (\Un_{\hh_S}\otimes Q_{\alpha, \beta}) F} \\
&=\tr{M^T Q_{\alpha, \beta}}\\
&=\tr{M Q_{\alpha, \beta}^T}~.
\end{align*}
Since $M^T=\Un-M$ we have $\lambda=\frac{1}{2}\tr{M Q_{\alpha, \beta}^T+M^T Q_{\alpha, \beta}}=\frac{1}{2}\tr{M(Q_{\alpha, \beta}^T-Q_{\alpha, \beta})}+\frac{1}{2}\tr{Q_{\alpha, \beta}}$, which is the formula of the theorem.

Let us derive an equation for $M$, $\tr{\call_\alpha^*(\rho) F F^*}=\lambda M$.
By Formula \ref{eq:cov_evolution} we know that
\[
\tra_{\yy}\left(\call^*(\rho)FF^*\right) =\left(-iT_S-\frac{1}{2} \Theta \Theta^*\right) M+M\left(iT_S-\frac{1}{2}\Theta \Theta^*\right)+\Theta M_B \Theta^*~.
\]
Thus we only need to compute $\tr{(\call_\alpha^*-\call^*)(\rho)FF^*}$. For any $i,j\in \set{1, \cdots, 2L_S}$, we have
\begin{align*}
\tr{(\call_\alpha^*-\call^*)(\rho)\gamma_i\gamma_j}&=\sum_{i=1}^n \tr{\rho \left(\call-\call_\alpha\right))(\gamma_i\gamma_j)}\\
&= \sum_{1\leq k,l\leq 2L_S} [Q_{\alpha, \beta}]^f_{k,l} \tr{\rho \gamma_k \gamma_i \gamma_j \gamma_l} \\
&= \sum_{1\leq k,l\leq 2L_S} [Q_{\alpha, \beta}]^f_{k,l} \left( [M]^f_{k,i} [M]^f_{j,l}-[M]^f_{k,j} [M]^f_{i,l}+[M_S]^f_{k,l}[M_S]^f_{i,j} \right)~~~\text{by the Wick formula} \\
&= [M^T Q_{\alpha, \beta} M^T-M Q_{\alpha, \beta}^T M+\tr{Q_{\alpha, \beta}^T M}M]^f_{i,j}~.
\end{align*}
Since $M^T=\Un-M$ we obtain
\begin{align*}
\tr{(\call_\alpha^*-\call^*)(\rho)F F^*}&=Q_{\alpha, \beta}-M Q_{\alpha, \beta}-Q_{\alpha, \beta} M+M\left(Q_{\alpha, \beta}-Q_{\alpha, \beta}^T\right) M+\tr{Q_{\alpha, \beta}^T M} M~.
\end{align*}
Finally, 
\begin{align}
\tr{\call_\alpha^*(\rho)F F^*}=G_\alpha M+M G_\alpha^* +M C_{\alpha, \beta} M+B_{\alpha, \beta}+\tr{Q_{\alpha, \beta}^T M} M ~
\end{align}
where $G_\alpha, B_{\alpha, \beta}, C_{\alpha, \beta}$ are defined in the theorem. Since $\lambda =\tr{Q_{\alpha, \beta}^T M}$, we have
\begin{align}
G_\alpha M+M G_\alpha^* +M C_{\alpha, \beta} M+B_{\alpha, \beta}=0~.
\end{align}
\end{proof}

We now turn to the study of the Riccati equation and the proof of Theorem \ref{theo:e_alpha_fermionic}. The problem is to find the solution of \ref{eq:riccati} for which the eigenvalue is maximal. We will need the following properties of Riccati equations:

\begin{prop}\label{prop:riccati_solution}
Let us consider some matrices $A, B, Q$ on $\C^d$ and consider the equation
\begin{align}\label{eq:riccati_gen}
XA+A^* X + XBX +C=0~.
\end{align}
Assume that $B$ and $C$ are self-adjoint, that $B\geq 0$, and that the Kalman space $K(A, B)$ is equal to $\C^d$ (the Kalman space is defined in Theorem \ref{theo:kalman}; we say that the pair $(A,B)$ is controllable).

Write $Z$ the matrix
\[
Z=\begin{pmatrix} A & B \\ -C & -A^*\end{pmatrix}~.
\] 
 Then for any solution $X$ of the Riccati equation the matrix $A+BX$ has for eigenvalues a subset of cardinal $d$ of the set of eigenvalues of $Z$ (counted with algebraic multiplicities).

Moreover, if Equation \eqref{eq:riccati_gen} admits a self-adjoint solution, then there is a self-adjoint solution $X_{max}$ such that $X\leq X_{max}$ for any self-adjoint solution $X$. The maximal solution $X_{max}$ is the unique self-adjoint solution whose eigenvalues are the eigenvalues of $Z$ with positive real part (counted with algebraic multiplicities). The maximal solution is isolated in the set of self-adjoint solutions.
\end{prop}
	
This is extracted from results scattered in \cite{Lancaster_Rodman95}. The fact that the maximal eigenvalue is isolated comes from Theorem 7.7.2. 
\vspace{0.5cm}

To convert Equation \eqref{eq:riccati} to an equation satisfying the hypothesis of this proposition, we note that $\rho_\alpha>0$  (since $\call_\alpha^*$ is positivity improving) so $M$ is of the form $(1+\exp(-T))^{-1}$ for some operator $T$.  We define
\[
X=M^{-1}-\Un~.
\]
We have $X>0$. Moreover,
\begin{align}\label{eq:riccatiX}
X A_\beta+A_\beta^* X +X B_{\alpha, \beta} X -B_{-\alpha, -\beta}=0~.
\end{align}
This formula is obtained by making the product of Equation \ref{eq:riccati} with $M^{-1}$ on the left and the right, and using the relations
\begin{align*}
A_\beta&=G_{\alpha, \beta}+B_{\alpha, \beta}\\ 
-B_{-\alpha, -\beta}&=C_{\alpha, \beta}+B_{\alpha, \beta}+G_{\alpha, \beta}+G_{\alpha, \beta}^*~.
\end{align*}

The equation on $X$ satisfies the hypothesis of Proposition \ref{prop:riccati_solution}. Indeed, the operator $B_{\alpha, \beta} =\sum_{i=1}^n e^{\alpha_i K_S}M_{\beta_i} \Theta_i \Theta_i^*$ is positive; since the semigroup is positivity improving the Kalman space $K(T_S, \Theta)$ is equal to $\yy$. But $A_\beta=-iT_S+R$ where $\ran(R)=\ran(\Theta)$, and $\ran(B_{\alpha, \beta})=\ran(\Theta)$, so $K(A_\beta, B_{\alpha, \beta})=\yy$.

Let us express $e(\alpha)$ in terms of $A_\beta +XB_{\alpha, \beta}$. We have
\[
-C_{\alpha, \beta}M=-M^{-1} M C_{\alpha, \beta} M=M^{-1}\left(G_{\alpha, \beta} M + M G_{\alpha, \beta}^* +B_{\alpha, \beta}\right)
\]
thus
\begin{align*}
e(\alpha)&=\frac{1}{2}\left(-\tr{C_{\alpha, \beta} M}+\tr{Q_{\alpha, \beta}}  \right) \\
&=\frac{1}{2}\tr{M ^{-1}G_{\alpha, \beta}M+G_{\alpha, \beta}^*+Q_{\alpha, \beta} +M^{-1} B_{\alpha, \beta}}\\
&=\frac{1}{2}\tr{G_{\alpha, \beta}+G^*_{\alpha, \beta}+Q_{\alpha, \beta}+B_{\alpha, \beta}+X B_{\alpha, \beta}} \\
&=\frac{1}{2} \tr{A_\beta+X B_{\alpha, \beta}}-\frac{1}{4}\sum_{i=1}^n \tr{\Theta_i \Theta_i^*}~.
\end{align*}
The last equality is due to the fact that $G_{\alpha, \beta}+G^*_{\alpha, \beta}+Q_{\alpha, \beta}+B_{\alpha, \beta}=A_\beta+i T_S-\frac{1}{2}\sum_{i=1}^n \Theta_i \Theta_i^*$ and $\tr{T_S}=0$.

Let us show that $X$ is the maximal solution $X_{max}$ of the Riccati equation \eqref{eq:riccatiX} (then the eigenvalues of $A_\beta+X_{max} B_{\alpha, \beta}$ are the eigenvalues of $Z$ with positive real parts, and the theorem is proved). First, we have $\tr{X B_{\alpha, \beta}}\leq \tr{X_{max} B_{\alpha, \beta}}$ since $X\leq X_{max}$. Thus, it is sufficient to show that $M_{max}=(\Un+X_{max})^{-1}$ is the covariance matrix of a state. This is equivalent to
\begin{align*}
X_{max}>0~\\
 X_{max}^T=X_{max}^{-1}~.
\end{align*}

We have $X_{max}\geq X >0$ so the inequality is satisfied. Moreover, we have $B_{\alpha, \beta}^T=B_{-\alpha, -\beta}$ and $A_\beta^T=-A_\beta^*$ so the map $X\mapsto -(X^T)^{-1}$ preserves the set of solutions of Equation \eqref{eq:riccatiX}. Since this map is increasing for the matrix order it sends a maximal solution on a maximal solution and $-\left(X_{max}^T\right)^{-1}=X_{max}$.  Thus $(\Un+X_{max})^{-1}$ is a covariance matrix, and it corresponds to the dominant eigenvector of $\call_\alpha^*$. This proves Theorem \ref{theo:e_alpha_fermionic}.

\subsection{The example of the fermionic chain}

In this subsection we describe the rate function of the large deviations on the fermionic chain of Paragraph \ref{subsub:gauge_invariant}, which we compute numerically for different values of the length $L$. The rate function $I$ is a function of two variables, but as a consequence of the following lemma we can consider only one parameter. 

\begin{lem}
If there are two baths, for any $\alpha_1, \alpha_2$ we have
\[
e(\alpha_1, \alpha_2)=e(\alpha_1-\alpha_2, 0)~.
\]
Writing $\tilde{e}(\alpha)=e(\alpha, 0)$ for any $\alpha\in \R$, we have
\[
I(\zeta_1, \zeta_2)=\left\{\begin{array}{cc} 
-\underset{\alpha \in \R}{\inf} \left(\scal{\alpha, \zeta_1}-\tilde{e}(\alpha)\right)& \text{ if $\zeta_1+\zeta_2=0$} \\
+\infty & \text{ if $\zeta_1+\zeta_2\neq 0$}
\end{array} \right.
\]
\end{lem}
This lemma is a straightforward consequence of \eqref{eq:I_sum}.

We computed $I(\zeta, -\zeta)$ for the fermionic chain with $\theta_0=\theta_{L+1}=0$, and temperatures $\beta_0=1, \beta_{L+1}=0$, for chains of lengths $L$ from $2$ to $5$ (see Figure \ref{fig:rate_1_0})

\begin{figure}[!h]
\includegraphics[scale=0.3]{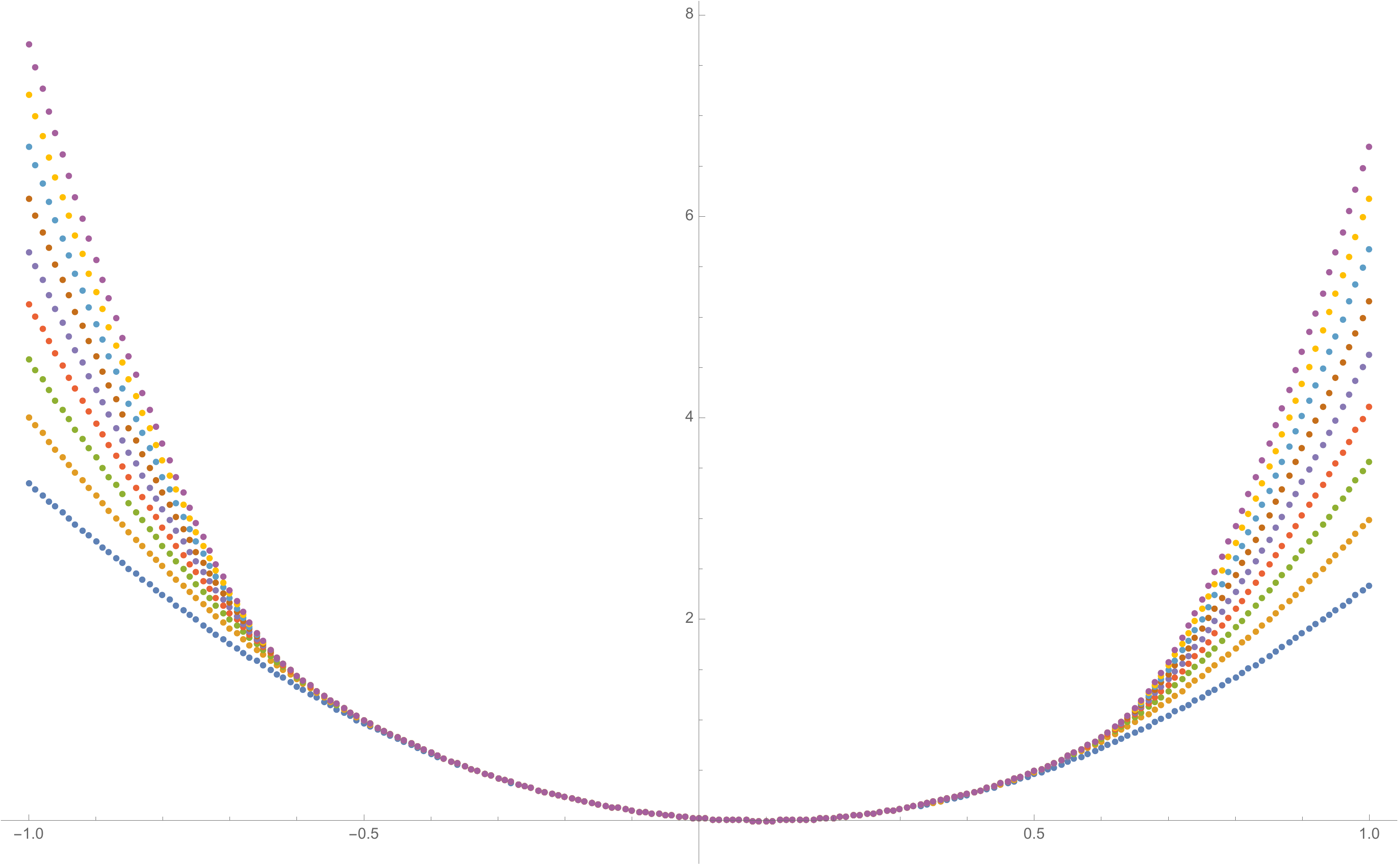}
\caption{Rate functional $I(\zeta, -\zeta)$ for $\beta_0=1, \beta_{L+1}=0$ and for $L=2$ to $L=10$. The largest function corresponds to $L=2$ and the smallest corresponds to $L=10$.}\label{fig:rate_1_0}
\end{figure}

As we can see, the rate functions have the same zero, which corresponds to the flux given in formula \eqref{eq:flux_chain}, that is 
\[
J_1=\frac{4}{10}(n_0-n_1)=\frac{4}{10}\left(\frac{1}{1+e^{-1}}-\frac{1}{1+e^0}\right)\simeq 0.092~.
\]

The rate functions are very similar around this zero, and progressively separate for large values of $\abs{\zeta-J_1}$. The rate function is smaller for large values of $L$, which means that the fluctuations of the energy fluxes around their mean values are larger when the length of the chain is larger. This result is interesting: the mean energy flux is completely independent of the length of the chain, but the large deviations are sensible to this length. 
\vspace{0.5cm}

In figure 2 we show the rate function in the case $\beta_0=10$ and $\beta_{L+1}=0$. Taking a high value of $\beta_0-\beta_{L+1}$ makes the asymmetry of $I$ under the change $\zeta\mapsto -\zeta$ very visible. 

\begin{figure}[!h]
\includegraphics[scale=0.3]{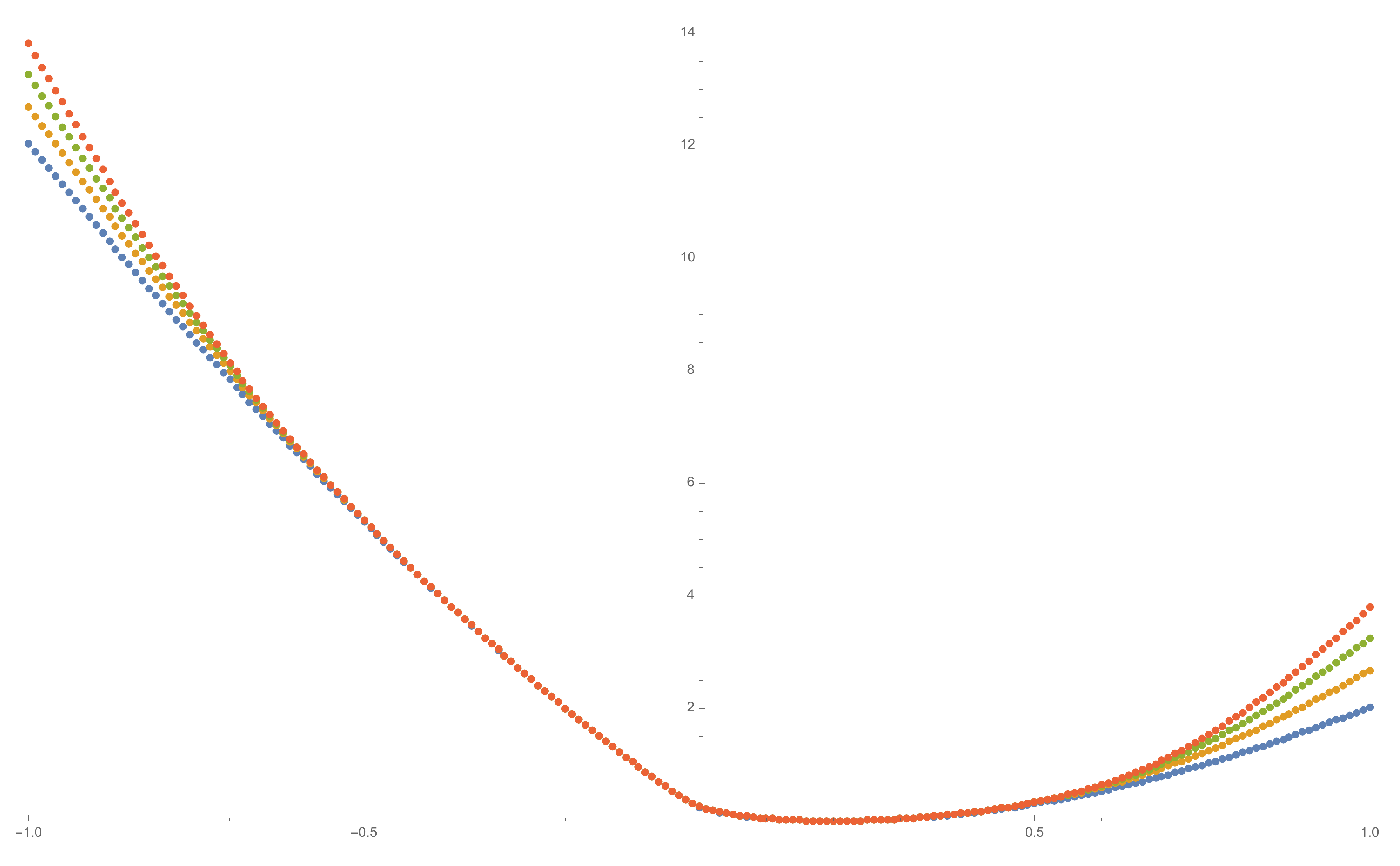}
\caption{Rate functional $I(\zeta, -\zeta)$ for $\beta_0=10, \beta_{L+1}=0$ and for $L=2$ to $L=5$. The largest function corresponds to $L=2$ and the smallest corresponds to $L=5$.}\label{fig:rate_10_0}
\end{figure}

\bibliography{/home/andreys/Documents/in_progress/bibli_quantum.bib}

\end{document}